\documentclass[journal,comsoc]{IEEEtran}
\usepackage{graphicx}
\usepackage{amssymb}
\usepackage{amsmath}
\usepackage{algorithmic}
\usepackage{algorithm}
\usepackage{array}
\usepackage{subfig}
\usepackage{textcomp}
\usepackage{physics}
\usepackage{stfloats}
\usepackage{url}
\usepackage{verbatim}
\usepackage{graphicx}
\usepackage{cite}
\usepackage[short, c3]{optidef}
\usepackage{epstopdf}
\usepackage[displaymath,mathlines]{lineno}
\usepackage{color}
\usepackage{soul}
\usepackage{tabulary}
\usepackage{multirow}
\usepackage{pbox}
\usepackage{multicol}
\usepackage{lipsum}
\usepackage{amsthm}
\usepackage{relsize}
\usepackage{lipsum}
\usepackage{epstopdf}
\usepackage{mathtools}
\usepackage{calc}
\usepackage{makecell}
\usepackage{url}
\usepackage{pgfplots}
\usepackage{subfig}

\newtheorem{theorem}{Theorem}
\newtheorem{lemma}{Lemma}
\newtheorem{corollary}{Corollary}

\newtheorem{example}{Example}
\newtheorem{remark}{Remark}

\begin{document}

\title{Active and Passive Beamforming Designs for SER Minimization in RIS-Assisted MIMO Systems}

\author{Trinh Van Chien,~\IEEEmembership{Member,~IEEE}, Bui Trong Duc, %
Ho Viet Duc Luong, Huynh Thi Thanh Binh,~\IEEEmembership{Member,~IEEE}, Hien Quoc Ngo, \textit{Senior Member,~IEEE}, and Symeon Chatzinotas, \textit{Fellow,~IEEE} \vspace{-0.5cm}
\thanks{Trinh Van Chien, Bui Trong Duc, Ho Viet Duc Luong, and Huynh Thi Thanh Binh are with the School of Information and Communication Technology (SoICT), Hanoi University of Science and Technology (HUST), Hanoi 100000, Vietnam (email: chientv@soict.hust.edu.vn, buitrongduc0502tlhpvn@gmail.com, luong.hvd242032M@sis.hust.edu.vn, binhht@soict.hust.edu.vn). Hien Quoc Ngo is with the School of Electronics, Electrical Engineering and Computer Science, Queen's University Belfast, Belfast BT7 1NN, United Kingdom (email: hien.ngo@qub.ac.uk). Symeon Chatzinotas is with the Interdisciplinary Centre for Security, Reliability and Trust (SnT), University of Luxembourg, L-1855 Luxembourg, Luxembourg (email: symeon.chatzinotas@uni.lu). (\textit{Corresponding author: Trinh Van Chien}). Parts of this paper has been accepted to present at the Genetic and Evolutionary Computation Conference (GECCO), 2024 \cite{Chien2024gecco}.}
\thanks{S. Chatzinotas was supported by the Luxembourg
National Fund (FNR)-RISOTTI–the Reconfigurable Intelligent Surfaces for Smart Cities under Project FNR/C20/IS/14773976/RISOTTI. The work of H. Q. Ngo was supported by the U.K. Research
and Innovation Future Leaders Fellowships under Grant MR/X010635/1. Ho Viet Duc Luong was funded by the Master, PhD Scholarship Programme of Vingroup Innovation Foundation (VINIF), code VINIF.2024.ThS.32. This work was sponsored by the Office of Naval Research (ONR) Global (ONRG) and the U.S. Army Combat Capabilities Development Command Indo-Pacific (DEVCOM Indo-Pacific), under grant number N62909-23-1-2018. The views and conclusions contained herein are those of the authors only and should not be interpreted as representing those of the U.S. Government.}}



\maketitle

\begin{abstract}
This research exploits the applications of reconfigurable intelligent surface (RIS)-assisted multiple input multiple output (MIMO) systems, specifically addressing the enhancement of communication reliability with modulated signals. Specifically, we first derive the analytical downlink symbol error rate (SER) of each user as a multivariate function of both the phase-shift and beamforming vectors. The analytical SER enables us to obtain insights into the synergistic dynamics between the RIS  and MIMO communication. We then introduce a novel average SER minimization problem subject to the practical constraints of the transmitted power budget and phase shift coefficients, which is NP-hard. By incorporating the differential evolution (DE) algorithm as a pivotal tool for optimizing the intricate active and passive beamforming variables in RIS-assisted communication systems, the non-convexity of the considered SER optimization problem can be effectively handled. Furthermore, an efficient local search is incorporated into the DE algorithm to overcome the local optimum, and hence offer low SER and high communication reliability. Monte Carlo simulations validate the analytical results and the proposed optimization framework, indicating that the joint active and passive beamforming design is superior to the other benchmarks.
\end{abstract}

\begin{IEEEkeywords}
RIS-assisted MIMO, beamforming design, communication reliability,  differential evolution, {SER minimization}.
\end{IEEEkeywords}
\vspace{-0.25cm}
\section{Introduction}
The field of sixth-generation (6G) communication stands at the forefront of contemporary scientific research, heralding a new era of ubiquitous connectivity. Advanced technologies are navigated in the intricate realms of future networks to grapple with challenges and envision innovative solutions to propel global communication beyond the current capabilities. From unprecedented spectral efficiency to ultra-low latency, 6G aims to revolutionize the way radio networks perceive and engage with the digital landscape. Both industry and academia have explored cutting-edge technologies such as terahertz frequencies \cite{hillger2019terahertz,song2022terahertz,shafie2023terahertz}, advanced artificial intelligence \cite{jagannath2021redefining,ahammed2023vision,ji2021survey}, and seamless integration of heterogeneous networks 
to unlock the full potential for the future. Wireless communications witness transformative shifts with each passing decade, marked by an evolution of hardware and software to boost spectral and energy efficiency \cite{jain2023dynamic,van2020joint,malik2021energy}.  
As the scientific community collaborates on this frontier, the pursuit of efficient, reliable, and ubiquitous connectivity remains paramount, shaping the trajectory of future communication networks.
{In a rapid growth of data rate demands, a strategic collaboration between reconfigurable intelligent surfaces (RISs) and multiple-input multiple-output (MIMO) \cite{van2021reconfigurable} creates a revolutionary transformation in the very fabric of beyond 5G networks. This combination addresses the challenges of propagation environments, ushering in a new era of adaptability, efficiency, and resilience. 
The heart of a RIS is a passive array of scattering elements strategically positioned within the coverage area  \cite{wu2019intelligent,wu2021intelligent}. This intelligent surface dynamically manipulates the phase and amplitude of incoming signals, and, in other words, acts like a ``smart mirror" for constructive combinations at the receivers \cite{elmossallamy2020reconfigurable,pan2021reconfigurable}. 
The RIS-assisted MIMO systems have demonstrated substantial gains in improving throughput, extending coverage, and optimizing energy efficiency toward the 6G networks.}

{For resource allocation in RIS-assisted MIMO systems, the optimization process involves dynamic adjustments of transmit power, beamforming, and time-frequency slots. The introduction of RIS adds a layer of complexity originating from the intelligent manipulation of signal characteristics. Power allocation was obtained across antennas and RIS elements, focusing on maximizing capacity and energy efficiency \cite{luo2021spatial,you2020energy}. Intelligent beamforming, facilitated by controlled phase shifts, was optimized to direct signals with precision, adapting to changing channel conditions with the presence of an RIS \cite{van2023long,van2023phase}. Time-frequency slots were efficiently scheduled under real-time feedback mechanisms for accurate channel state information (CSI) \cite{jiang2023capacity}. Dynamic user association and mobility management were optimized in \cite{bie2023user, shah2022statistical} to enhance system adaptability subject to the quality of service requirements. The considered optimization problems in these previous works only include common mathematical expressions to exploit features such as the gradients to guide the solution.}

{Resource allocation is of paramount importance in radio networks to maximize spectral and energy efficiency \cite{you2020energy}. Communication reliability is a cornerstone in wireless systems represented by, for example, the outage probability \cite{nguyen2023dilemma, yang2020outage}, the communication robustness \cite{zheng2022robust}, and the symbol error rate (SER) \cite{ye2020joint} that plays a pivotal role in assessing the robustness of signal transmission.   
As a fundamental indicator of communication performance, achieving low SER for each user necessitates the implementation of effective error control mechanisms, modulation schemes, and sophisticated channel coding strategies aimed at minimizing the impact of disturbances on symbol accuracy. 
These designs underscored the continuous efforts in research, which were exemplified by \cite{law2018symbol,basar2020reconfigurable,ye2020joint}, to enhance communication reliability as a crucial metric for ensuring resilience in fading channels. Nonetheless, the SER is expressed by advanced functions so that the error probability optimization is nontrivial to solve for RIS-assisted networks even with single-user scenarios \cite{basar2020reconfigurable}, where the decoding error probability is non-neglected.}

{From an algorithmic standpoint, evolutionary computing is a field that explores efficient search algorithms rooted in Darwin's theory of evolution \cite{li2022evolutionary}. These metaheuristic methods are collectively referred to as evolutionary algorithms (EAs) encompassing population-based stochastics renowned for their speed, efficiency, and capacity to tackle a diverse range of complicated optimization problems to seek extrema within multivariate, nonlinear, non-differentiable, or multi-modal objective and constraints. A volume of work has consistently demonstrated that EAs yield high-quality solutions for computationally challenging problems classified as NP-hard or NP-complete as in \cite{binh2023ensemble} and reference therein. For 6G communications, EAs have manifested their effectiveness in featuring a multitude of resource allocation problems, with a pronounced relevance in the context of resilient networks  \cite{van2023phase, lipare2021fuzzy, xue2021routing}.
In \cite{van2023phase},  the efficacy of the differential evolution (DE) algorithm was proposed to maximize the total throughput within a complex objective function characterized by spatial correlation fading channels. Cloud optimization elucidated in \cite{lipare2021fuzzy} was proven as a potent approach for the selection of cluster centers by particle swarm optimization (PSO). The application of genetic algorithms (GA) in \cite{xue2021routing} was instrumental in addressing issues related to uncovered areas in network routing. Despite a plethora of substantial applications for 6G communications, there is room to design effective algorithms since EAs may be trapped in locally optimal solutions.}

In this paper, we study the {SER minimization problem} of modulated RIS-assisted MIMO systems under the active and passive beamforming design to simultaneously serve multiple users. Regarding the performance analysis, we derive the SER of each user of multi-user communication systems instead of single-user scenarios \cite{law2018symbol,basar2020reconfigurable,ye2020joint}. Furthermore, we formulate and solve the total SER minimization that may obtain a better solution than the state-of-the-art benchmarks \cite{peng2021analysis, huang2022placement}. Our main contributions are sketched as follows
\begin{itemize}
\item We derive the closed-form expression of the SER for each user in a downlink data transmission with the quadrature amplitude modulation (QAM) by effectively treating mutual interference as additive noise. The analytical results reveal the influence of both the beamforming vectors and phase shift coefficients to attain low SER for the users. 
\item We formulate a general joint active and passive beamforming design to minimize the average SER for the modulated communication systems, which is a non-convex and NP-hard optimization problem. The phase shift designs for the systems utilizing linear signal processing are particular exhibitions of our consideration. 
\item We encode the optimization variables into the individual that can map to the DE framework. An improved DE algorithm with local search is then proposed to evolve the individuals escaping local optimums along generations. The convergence of our proposed algorithm is also mathematically obtained. 
\item Numerical results manifest the effectiveness of the analytical SERs. Moreover, the passive and active beamforming designs demonstrate the superiorities in {minimizing SER of users} for intelligent networks. 
\end{itemize}
The rest of this paper is organized as follows: Section~\ref{Sec:Sys} introduces the considered multi-user RIS-assisted MIMO system model and derives the downlink SER of every user in the coverage area. In Section~\ref{Sec:Opt}, we formulate a communication reliability optimization problem that minimizes the average SER by considering the beamforming vectors and the scattering elements as optimization variables. We further propose an improved DE algorithm to solve the SER minimization problem in Section~\ref{Sec:DE}. Extensive numerical results are presented in Section~\ref{Sec:NumRe}, while main conclusions are drawn in Section~\ref{Sec:Concl}. 

\textit{Notation}:  Lower bold and upper bold letters denote vectors and matrices, respectively. The notation $\|\cdot\|$ stands for the Euclidean norm and $\lfloor \cdot\rfloor $ denotes the floor function. Regular and Hermitian transposes are denoted by the subscripts $(\cdot)^T$ and $(\cdot)^H$. The identity matrix of size $N \times N$ is denoted as $\mathbf{I}_N$ and $\mathrm{diag}(\mathbf{x})$ a diagonal matrix with the elements of the vector $\mathbf{x}$ on the diagonal. The expectation of a random variable is $\mathbb{E}\{ \cdot \}$. Meanwhile, $\mathcal{N}(\cdot, \cdot)$ and $\mathcal{CN} (\cdot, \cdot)$ represent the Gaussian and circularly symmetric complex Gaussian distribution, respectively. $\mathsf{Pr}(\cdot)$ denotes the probability of an event.  Finally, $\Re(\cdot)$ and $\Im(\cdot)$ denote the real and imaginary parts of a complex number.
\vspace{-0.25cm}
\section{RIS-Assisted System Model and Downlink Symbol Error Rate} \label{Sec:Sys}
\begin{figure}
    \centering
    \includegraphics[width=2.1in]{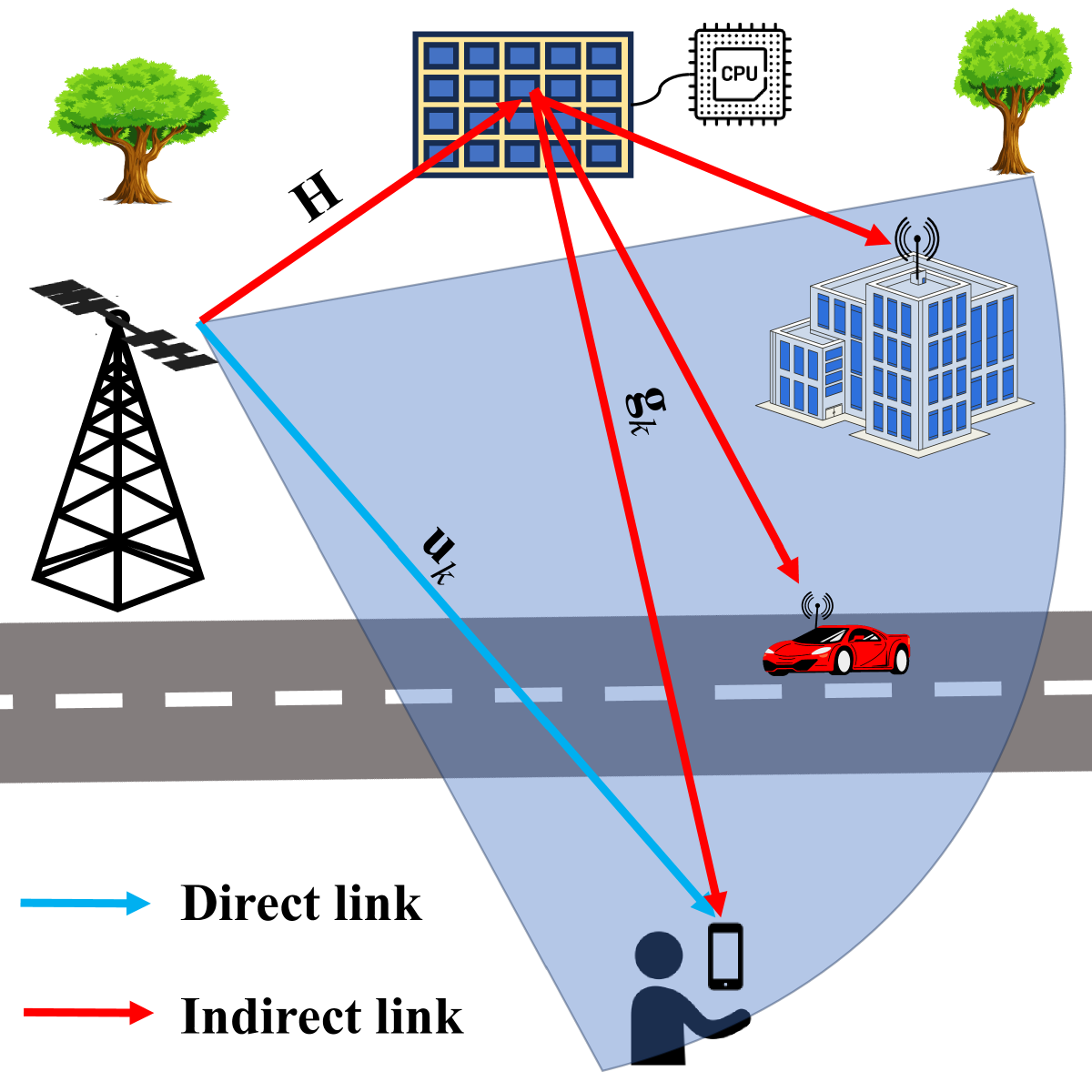}
    \caption{The considered RIS-assisted MIMO system model where a BS serves multiple users.}
    \label{FigSysModel}
    \vspace{-0.5cm}
\end{figure}
{We consider a RIS-assisted MIMO system in which a BS is equipped with $M$ antennas serving $K$ single-antenna users randomly distributed in the coverage area as illustrated in Fig.~\ref{FigSysModel}.\footnote{{A scenario where one base station equipped with multiple antennas served a single-antenna is multiple-input single-output (MISO). Aligned with previous work, we have utilized MIMO systems in this paper since the base station serves multiple users. Even though each active beamforming are a matrix as users having multiple antennas, our proposal evolutionary solution can be applied with efforts in adapting individuals and fine tuning parameters.}} The system performance is supported by a reconfigurable intelligent surface (RIS) consisting of $N$ scattering elements that are capable of modifying the phase of incoming signals. The phase shift matrix of the RIS  is denoted as $\pmb{\Phi} = \mathrm{diag} \left([ e^{j\theta_1}, \ldots, e^{j\theta_N}]^T \right)$, characterizing the scattering properties of the $N$ elements with $j$ being the complex component, i.e., $j^2 = -1$. In particular, $\theta_n$ represents the phase shift introduced by the $n$-th element within the range of $[-\pi, \pi]$.\footnote{{The phase shifts are typically discrete in practice. Different from deterministic optimization, our proposed algorithm can be extended to apply for such practical applications since the fitness function is independent of the first and second derivatives of the objective function and constraints. Due to the discrete feasible region, the adaption is required and it should be of interest for future work.}} Besides, a central processing unit (CPU) located at, for example, the BS, coordinates the activities of the RIS and BS to cater to users via sharing the same time and frequency resources. Let us denote $\mathbf{H} \in \mathbb{C}^{M \times N}$ the channel matrix between the BS and the RIS, while $\mathbf{g}_k \in \mathbb{C}^{N}$ denotes the channel vector between the RIS and user $k$. For the direct link, the channel between the BS and user~$k$ is $\mathbf{u}_k \in \mathbb{C}^{M}$.} 
\vspace{-0.25cm}
\subsection{Downlink Data Transmission}
{In the downlink data transmission,  all the users simultaneously receive signals from the BS with the assistance of the RIS. Let us denote $s_k$ with $\mathbb{E}\{\| s_k\|_2^2 \} = 1$ the data symbol transmitted from the BS to user~$k$. This data symbol is steered to user~$k$ by a beamforming vector $\mathbf{w}_k \in \mathbb{C}^M$ with $\|\mathbf{w}_k\|_2 =1$. Consequently, the received signal at user~$k$, denoted by $y_k \in \mathbb{C}$, is formulated as
\begin{equation} \label{eq:yk}
\begin{split}
    y_k & = \mathbf{u}_k^H \sum_{k' = 1}^{K} \sqrt{\rho} \mathbf{w}_{k'} s_{k'} + \mathbf{g}_k^H \pmb{\Phi}^H \mathbf{H}^H \sum_{k' = 1}^{K} \sqrt{\rho} \mathbf{w}_{k'} s_{k'} + n_k\\
    & = (\mathbf{u}_k + \mathbf{H}\pmb{\Phi}\mathbf{g}_k)^H \sum_{k' = 1}^K \sqrt{\rho} \mathbf{w}_{k'} s_{k'} + n_k,
\end{split}
\end{equation}
where $\rho > 0$ is the transmit power allocated to each data symbol, which is uniformly allocated across users. $n_k \sim \mathcal{CN}(0, \sigma^2 )$ is additive noise with zero mean and variance $\sigma^2$. Let us introduce a new variable $\mathbf{z}_k = \mathbf{u}_k + \mathbf{H}\pmb{\Phi}\mathbf{g}_k $ that is the aggregated channel comprising both the indirect and direct links. After that, the received signal in \eqref{eq:yk} is equivalent to 
\begin{equation} \label{eq:ykv1}
    y_k = \mathbf{z}_k^H \sqrt{\rho} \sum_{k' = 1}^K   \mathbf{w}_{k'} s_{k'} + n_k,
\end{equation}
In order to decode the data symbols sent from user~$k$, \eqref{eq:ykv1} is decomposed to as
\begin{equation} \label{eq:ykv2}
    y_k  = \sqrt{\rho} \mathbf{z}_k^H\mathbf{w}_k s_k + \sqrt{\rho} \sum_{k' =1,k' \neq k }^K \mathbf{z}_k^H \mathbf{w}_{k'}s_{k'} + n_k,
\end{equation}
where the first part in \eqref{eq:ykv2} contains the desired signal from user~$k$, while the second part in \eqref{eq:ykv2} is mutual interference from the other users. By utilizing the cofficient $\sqrt{\rho} \mathbf{z}_k^H \mathbf{w}_k$ for equalization, the  symbol interested by user~$k$ is defined from
\begin{equation} \label{eq:rk}
    r_k = s_k + \sum_{k' = 1, k'\neq k}^K \frac{\mathbf{z}_{k}^H \mathbf{w}_{k'}}{\mathbf{z}_k^H \mathbf{w}_k} s_{k'} + \frac{n_k}{\sqrt{\rho}\mathbf{z}_k^H \mathbf{w}_k}.
\end{equation}
We stress that $r_k$ is different from $s_k$ in general due to nonneglectable interference and noise. However, user~$k$ can still decode the data symbol correctly by using, for example, the maximum likelihood detection if $r_k$ falls in the Voronoi region \cite{aurenhammer1991voronoi} of the original data symbol $s_k$. The SER of user~$k$, denoted by $\mathsf{SER}_k$, is computed as
\begin{equation} \label{eq:SERk}
\mathsf{SER}_k (\{\mathbf{w}_k\}, \pmb{\Phi}) = \mathsf{Pr}(r_k \notin \mathcal{V}(s_k)),
\end{equation}
where 
$\mathcal{V}(s_k)$ is the Voronoi region of the constellation point represented for $s_k$.  We now investigate  $m$-QAM with $m$ being the modulation index, to modulate the binary data sequence before transmitting the signals over the medium. In $m$-QAM, the modulation constellation set  $\mathcal{M}$ consists of $m$ points, which are defined as
\begin{equation} \label{eq:M}
    \mathcal{M} = \left\{\pm \left(2p-1\right) \delta \pm j\left(2q-1\right)\delta \right\},
\end{equation}
where $p,\ q \in \left\{1,\ldots, \frac{\sqrt{m}}{2}\right\}$.  In order to derive the analytical SER expression,  the set of the constellation points illustrated in Fig.~\ref{Constellation} is classified into three subsets based on their Voronoi regions, say $\mathcal{S}_1, \mathcal{S}_2$ and $\mathcal{S}_3$. The three subsets are:
\begin{itemize}
\item[$i)$] \textit{Subset $\mathcal{S}_1$} involves four corner constellation points defined as $\pm \delta(2p-1) \pm j\delta(2q-1)$ with $p= q = \sqrt{m}/2$. For example the constellation point $s_1 = \delta(\sqrt{m}-1) + j\delta(\sqrt{m}-1)$ has the Voronoi region defined as
    \begin{multline}\label{eq:voronois_1}
        \mathcal{V}(s_1)=\big\{s\in \mathbb{C} \ \big|  \Re(s) \geq \delta(\sqrt{m}-2) \\
        \&  \ \Im(s) \ge \delta(\sqrt{m}-2) \big\}.
    \end{multline}
\item[$ii)$] \textit{Subset $\mathcal{S}_2$} involves $4(\sqrt{m}-2)$ boundary points (excluding the corner points in $\mathcal{S}_1$), which are $\pm \delta(2p-1) \pm j\delta(2q-1)$ with either $p = \sqrt{m} /2 $ and $ q\in \{1, \ldots, \sqrt{m}/2-1\}$ or $q = \sqrt{m}/2$ and $p\in \{1, \ldots, \sqrt{m}/2-1\}$. The Voronoi region of a specific point $s_2 =  \delta(\sqrt{m}-1) + j\delta(\sqrt{m}-3)$ is
    \begin{multline}
        \mathcal{V}(s_2)=\big\{s \in \mathbb{C} \ \big| \Re(s) \ge \delta(\sqrt{m}-2)  \\ 
        \& \,\  \delta(\sqrt{m}-4) \le  \Im(s) < \delta(\sqrt{m}-2)  \big\}.
    \end{multline}
\item[$iii)$] \textit{Subset $\mathcal{S}_3$} involves $(\sqrt{m} - 2)^2$ interior points defined as $\pm \delta(2p-1) \pm j\delta(2q-1)$, where $p, q \in \{1,\ldots, \sqrt{m}/2-1 \}$. An example point in this category is $s_3 =  \delta(\sqrt{m}-3) + j\delta(\sqrt{m}-3)$ with its Voronoi region defined as
    \begin{multline}
        \mathcal{V}(s_3)=\big\{s \in \mathbb{C} \ \big|\  \delta(\sqrt{m}-4) \le  \Re(s) < \delta(\sqrt{m}-2)   \\ 
         \&\,\  \delta(\sqrt{m}-4) \le  \Im(s) < \delta(\sqrt{m}-2)  \big\}.
    \end{multline}
\end{itemize}
Fig.~\ref{Constellation} shows the 16-QAM in which the Voronoi region of each constellation point is well-established. The SER in \eqref{eq:SERk} can be evaluated numerically for arbitrary channel models and beamforming vectors. Nonetheless, it requires a very long data sequence to achieve a high accuracy of the SER, which is burdensome for large-scale networks.  
\vspace{-0.25cm}
\subsection{Analysis of Symbol Error Rate}
To go further insight into the analysis of the error probability, we set $\zeta_k$ and $\nu_k$ as follows\footnote{In this paper, we focus on multiple-access scenarios where the system simultaneously serves a number of users with the same time and frequency resource. The SER is, therefore, a function of both mutual interference and noise by treating the mutual interference as effective noise. The multiple BSs and RISs are potential research directions, but the extensions are nontrivial. A concrete protocol to exchange information between the BSs and RISs should be set up. It is beyond the scope of this paper and we leave this potential extension for future work.}
\begin{equation} \label{eq:Zknuk}
\zeta_k = \sum_{k' = 1, k'\neq k}^K \frac{\mathbf{z}_{k'}^H \mathbf{w}_{k'}}{\mathbf{z}_k^H \mathbf{w}_k} s_{k'}, \nu_k = \frac{n_k}{\sqrt{\rho}\mathbf{z}_k^H \mathbf{w}_k},
\end{equation} 
and treat $\zeta_k$ as noise. Then conditioned on channel gains and under the assumption of Gaussian signaling, these random variables are distributed as 
\begin{equation} \label{eq:INkNOk}
\zeta_k \sim \mathcal{CN}\left(0, \mathsf{IN}_k(\{ \mathbf{w}_k \}, \pmb{\Phi})\right), \nu_k \sim \mathcal{CN}\left(0, \mathsf{NO}_k(\{ \mathbf{w}_k \}, \pmb{\Phi})\right),
\end{equation}
where the corresponding variances, $\mathsf{IN}_k(\{ \mathbf{w}_k \}, \pmb{\Phi})$ and $\mathsf{NO}_k(\{ \mathbf{w}_k \}, \pmb{\Phi})$, are defined as
\begin{equation} \label{eq:IkNk}
\begin{split}
        \mathsf{IN}_k(\{ \mathbf{w}_k \}, \pmb{\Phi}) = \sum_{\substack{k' = 1, \\ k'\neq k}}^K \left| \frac{\mathbf{z}_k^H \mathbf{w}_{k'}}{\mathbf{z}_k^H \mathbf{w}_k} \right|^2, 
        \mathsf{NO}_k(\{ \mathbf{w}_k \}, \pmb{\Phi}) =  \frac{\sigma^2}{\rho \left| \mathbf{z}_k^H \mathbf{w}_k \right|^2}.
\end{split}
\end{equation}
}

{Let us consider a scenario where the BS sends equiprobable symbols so that
$\mathsf{Pr}(s_k = s_t) = 1/m, \forall s_t \in \mathcal{M}$ using the $m$-QAM.
After that, the average of energy per symbol, denoted by $E_s$, is computed as
\begin{equation}
\begin{split}
   & E_s \stackrel{(a)}{=} \frac{1}{m} \sum_{s_t \in \mathcal{M}} E_{s_t}  \stackrel{(b)}{=}\frac{1}{m} \sum_{p=1}^{\sqrt{m}/2} \sum_{q=1}^{\sqrt{m}/2} 4\delta^2 ((2p-1)^2+(2q-1)^2)\\
    &\stackrel{(c)}{=} \frac{2\delta^2}{\sqrt{m}} \left( \sum_{p=1}^{\sqrt{m}/2} (2p-1)^2 +  \sum_{q=1}^{\sqrt{m}/2} (2q-1)^2 \right) \\
    &= \frac{4\delta^2}{\sqrt{m}}  \sum_{p=1}^{\sqrt{m}/2} (2p-1)^2 =\frac{4\delta^2} {\sqrt{m}} \left(4\sum_{p=1}^{\sqrt{m}/2} p^2 - 4\sum_{p=1}^{\sqrt{m}/2} p + \frac{\sqrt{m}}{2} \right)   \\ 
    &\stackrel{(d)}{=} \frac{4\delta^2}{\sqrt{m}} \left( \frac{\sqrt{m}(\sqrt{m}+2)(\sqrt{m}+1)}{6} - \frac{\sqrt{m}(\sqrt{m}+2)}{2} + \frac{\sqrt{m}}{2}
    \right) \\ 
    & = \frac{2}{3}\delta^2(m-1),
\end{split}
\end{equation}
where $(a)$ is because all the modulated symbols are equiprobable; $(b)$ is because for each $p, q \in \{1,\ldots, \sqrt{m}/2\}$ expressed in \eqref{eq:M}, there are four signal constellation points, i.e., $ \pm \delta(2p-1) \pm j \delta(2q-1)$, with the energy of $\delta^2((2p-1)^2+(2q-1)^2)$; $(c)$ is obtained by the identity $\sum_{y\in \mathcal{A}} g(x)  = |\mathcal{A}|g(x) $; and  $(d)$ is obtained by  the identities $\sum_{p=1}^n p = n(n+1)/2$ and $\sum_{p=1}^n p^2 = n(n+1)(2n+1)/6$. Hence, the half distance between the two nearest constellation points is defined as
\begin{equation}
   \delta = \sqrt{3 E_s/(2(m-1))}.
\end{equation}
Once the energy per symbol is normalized, i.e., $E_s =1$, we obtain the result in \eqref{eq:dmin}.  $\delta$ is half of the distance between the two nearest points on the constellation computed as
\begin{equation} \label{eq:dmin}
        \delta =\sqrt{3/(2(m-1))}.
\end{equation}}
\begin{figure}[t]
    \centering
    \includegraphics[width=3.0in, trim = 2.5cm 2.5cm 2.9cm 2.1cm, clip = true]{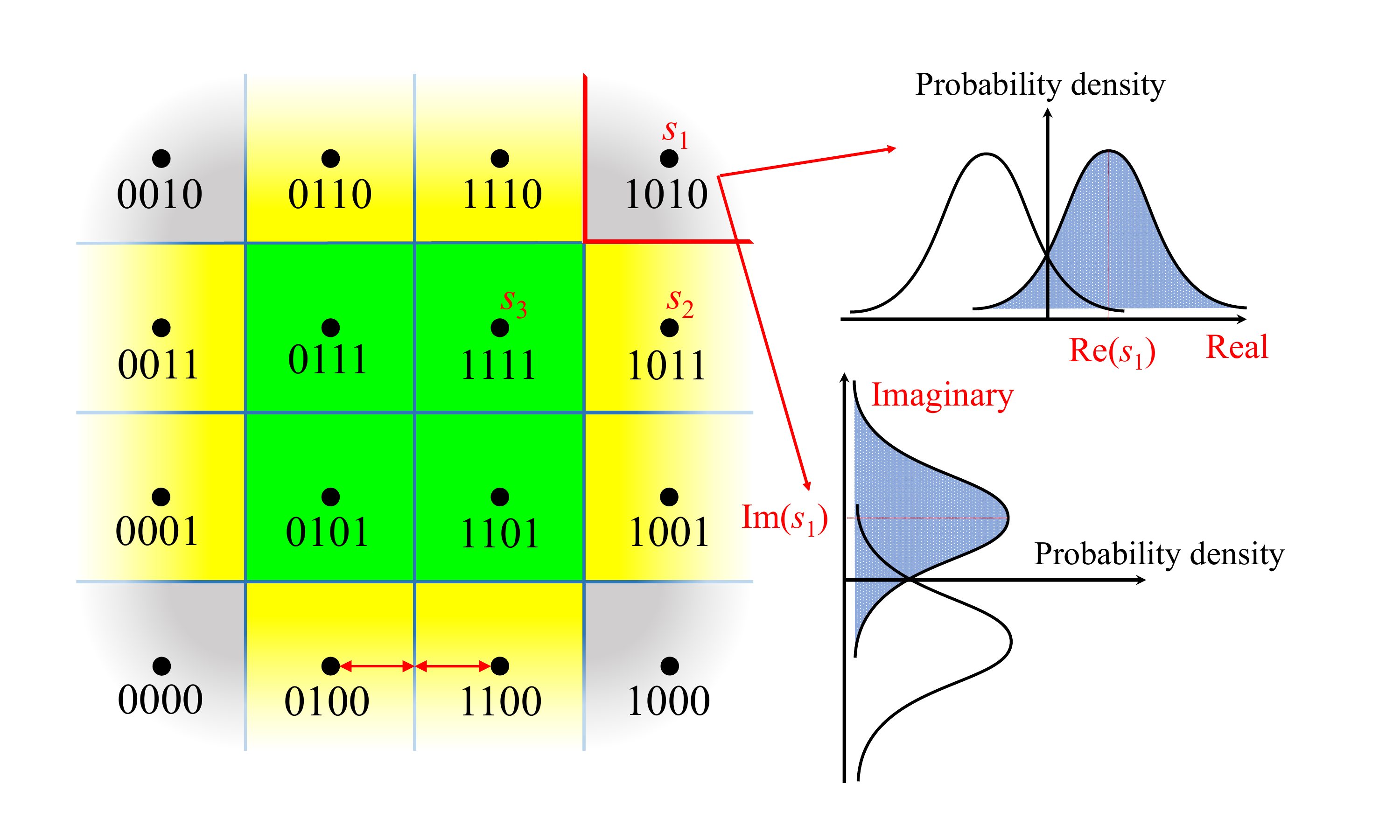}
    \caption{An example of 16-QAM with the Voronoi regions of the constellation points.}
    \label{Constellation}
    \vspace{-0.5cm}
\end{figure} 
{
As observed in \eqref{eq:rk} together with the assumption in \eqref{eq:INkNOk}, the SER of user~$k$ is subject to the presence of mutual interference expressed by an approximation of the Gaussian noise. In more detail, Gaussian noise can introduce random fluctuations in the amplitude and phase of the received signals, potentially contaminating the decoding quality. Note that, in multiple-user networks, a significant factor impacting the received signal is mutual interference originating from other users sharing the same time and frequency resources. It leads to the quality of the desired signals from one user being interfered with those of the remaining. Mutual interference can lead to signal distortion and a decrease in the signal-to-interference-and-noise ratio, affecting the reliability of a communication system as demonstrated in Lemma~\ref{theorem:ser}.\footnote{The SER expression for each user is obtained with perfect channel state information that can be considered as an upper bound in practice under low mobility environments. The communication reliability under imperfect channel state information is of interest for future work by exploiting, for example, the aggregated channel estimation \cite{van2022reconfigurable}.}}
\begin{lemma}\label{theorem:ser}
{The analytical downlink SER of user~$k$ is 
    \begin{multline}\label{eq:SER}
        \overline{\mathsf{SER}}_k (\{ \mathbf{w}_k \}, \pmb{\Phi}) =  2\left(1-\frac{1}{\sqrt{m}}\right)\mathrm{erfc}\left(\sqrt{\frac{3\mathsf{SINR}_k (\{ \mathbf{w}_k \}, \pmb{\Phi} )}{m-1}}\right) \\ 
        - \left(1-\frac{1}{\sqrt{m}}\right)^2\mathrm{erfc}^2\left(\sqrt{\frac{3\mathsf{SINR}_k(\{ \mathbf{w}_k \}, \pmb{\Phi})}{m-1}}\right),
    \end{multline}
    where the signal-to-interference-and-noise ratio of user~$k$, denoted by $\mathsf{SINR}_k(\{ \mathbf{w}_k \}, \pmb{\Phi})$, is driven based on the received signal in \eqref{eq:rk} as 
        \begin{equation}
            \begin{split}
            \label{eq:sinr}
            \mathsf{SINR}_k(\{ \mathbf{w}_k \}, \pmb{\Phi} ) 
            &= \frac{\rho \left| \mathbf{z}_k^H \mathbf{w}_k \right|^2}{\sum_{k' = 1, k'\neq k}^K  \rho \left| \mathbf{z}_k^H \mathbf{w}_{k'} \right|^2 +   \sigma^2},
        \end{split}
        \end{equation}
    and $\mathrm{erfc}(\cdot)$ is the complementary error function defined as}
       ${ \operatorname{erfc}(z) = 1 - \operatorname{erf}(z) =  
        \frac{2}{\sqrt{\pi}} \int_{z}^{\infty} e^{-t^2}\operatorname{d}t.}$
\end{lemma}
\begin{proof}
{The proof is accomplished based on computing the error probability of each data symbol over the fading channels using the Voronoi regions. The detailed proof is available in  Appendix~\ref{appen:ser}.}
\end{proof}
{The analytical SER obtained in Lemma~\ref{theorem:ser} is independent of the symbol stream and therefore reduces the computational complexity in evaluating the communication reliability. Moreover, \eqref{eq:SER} can be applied for an arbitrary beamforming technique and channel model. To seek further insight, we reformulate \eqref{eq:SER} into an equivalent expression by utilizing the series representation of the $\mathrm{erf}(\cdot)$ function \cite[3.321.1]{gradshteyn2014table} as
  \begin{multline}\label{eq:SERv1}
\overline{\mathsf{SER}}_k (\{ \mathbf{w}_k \}, \pmb{\Phi}) = \\
        2\left(1-\frac{1}{\sqrt{m}}\right)\left( 1- \frac{2}{\sqrt{\pi}}\sum_{l=0}^{\infty} \frac{(-1)^l 3^{l+0.5} \mathsf{SINR}_k^{l+0.5} (\{ \mathbf{w}_k \}, \pmb{\Phi} )}{l!(2l+1) (m-1)^{l+0.5}} \right)\\ 
        - \left(1-\frac{1}{\sqrt{m}}\right)^2\left( 1- \frac{2}{\sqrt{\pi}}\sum_{l=0}^{\infty} \frac{(-1)^l 3^{l+0.5} \mathsf{SINR}_k^{l+0.5} (\{ \mathbf{w}_k \}, \pmb{\Phi} )}{l!(2l+1) (m-1)^{l+0.5}} \right)^2,
    \end{multline}
which allows us to evaluate numerically the SER of user~$k$ with tolerable accuracy by truncating the summation. Specifically, at the low and moderate SINR regimes, one can approximate the SER as
\begin{equation}\label{eq:SERlow}
\begin{split}
&\overline{\mathsf{SER}}_k (\{ \mathbf{w}_k \}, \pmb{\Phi}) \approx 2\left(1-\frac{1}{\sqrt{m}}\right)
        \left( 1- \frac{2}{\sqrt{\pi}}\sqrt{\frac{ 3 \mathsf{SINR}_k (\{ \mathbf{w}_k \}, \pmb{\Phi} )}{ m-1 }} \right.\\ 
& \left. +  \frac{2}{\sqrt{\pi}}\sqrt{\frac{ 3\mathsf{SINR}_k^3 (\{ \mathbf{w}_k \}, \pmb{\Phi} )}{ (m-1)^3 }}\right) -  \left(1-\frac{1}{\sqrt{m}}\right)^2 \times \\
& \left( 1- \frac{2}{\sqrt{\pi}}\sqrt{\frac{ 3 \mathsf{SINR}_k (\{ \mathbf{w}_k \}, \pmb{\Phi} )}{ m-1 }} +  \frac{2}{\sqrt{\pi}}\sqrt{\frac{ 3 \mathsf{SINR}_k^3 (\{ \mathbf{w}_k \}, \pmb{\Phi} )}{ (m-1)^3 }}\right)^2.
\end{split}
\end{equation}
We emphasize that the approximation in \eqref{eq:SERlow} degrades the evaluation cost due to the simple algebra. For the high SINR regime, let first us borrow the series representation of the $\mathrm{erfc}(\cdot)$ function \cite[7.1.23]{abramowitz1964handbook}  to tackle the SER of user~$k$ as
\begin{equation}\label{eq:SERv2}
\begin{split}
&\overline{\mathsf{SER}}_k (\{ \mathbf{w}_k \}, \pmb{\Phi}) =  2\left(1-\frac{1}{\sqrt{m}}\right)  \frac{e^{\frac{-3 \mathsf{SINR}_k (\{ \mathbf{w}_k \}, \pmb{\Phi} )}{ m-1 }}\sqrt{m-1}}{\sqrt{ 3 \pi\mathsf{SINR}_k (\{ \mathbf{w}_k \}, \pmb{\Phi} )}} \\
&\times\sum_{l=0}^{\infty} \frac{(-1)^l  (2 l-1)!! (m-1)^l}{  6^l \mathsf{SINR}_k^l (\{ \mathbf{w}_k \}, \pmb{\Phi} )}
    - \left(1-\frac{1}{\sqrt{m}}\right)^2   \\
&\frac{e^{\frac{-6 \mathsf{SINR}_k (\{ \mathbf{w}_k \}, \pmb{\Phi} )}{ m-1 }}(m-1)}{3 \pi\mathsf{SINR}_k (\{ \mathbf{w}_k \}, \pmb{\Phi} )}  \left(\sum_{l=0}^{\infty} \frac{(-1)^l (2l-1)!! (m-1)^l}
    {6^l \mathsf{SINR}_k^l (\{ \mathbf{w}_k \}, \pmb{\Phi} )} \right)^2,
\end{split}
\end{equation}
where $(2i-1)!!$ is the double factorial and $(-1)!! = 1$. After that, by only taking the first-order approximation of \eqref{eq:SERv2}, i.e., $n=1$, the SER of user~$k$ can be simplified to as
\begin{multline}
        \overline{\mathsf{SER}}_k (\{ \mathbf{w}_k \}, \pmb{\Phi}) \approx 
        2\left(1-\frac{1}{\sqrt{m}}\right)   \frac{e^{\frac{-3 \mathsf{SINR}_k (\{ \mathbf{w}_k \}, \pmb{\Phi} )}{ m-1 }}\sqrt{m-1}}{\sqrt{ 3 \pi\mathsf{SINR}_k (\{ \mathbf{w}_k \}, \pmb{\Phi} )}}\\ 
        - \left(1-\frac{1}{\sqrt{m}}\right)^2 \frac{e^{\frac{-6 \mathsf{SINR}_k (\{ \mathbf{w}_k \}, \pmb{\Phi} )}{ m-1 }}(m-1)}{3 \pi\mathsf{SINR}_k (\{ \mathbf{w}_k \}, \pmb{\Phi} )},
\end{multline}
which can be evaluated efficiently. Consequently, the analytical SER in \eqref{eq:SER} offers a baseline for approximations with a lower computational cost by conditioning on the range of the SINR. Note that the SER in Theorem~\ref{theorem:ser} and its approximations can be minimized concerning both the beamforming vectors and the phase shift matrix. In practice, we can exploit linear beamforming vectors to reduce the system cost.}

\subsection{{Symbol Error Rate} Analysis with Linear Beamforming Techniques}
{We now consider scenarios where the system deploys a specific linear beamforming technique. In more detail,  the maximum ratio transmission (MRT), zero-forcing (ZF), and regularized zero-forcing (RZF) techniques are respectively considered in this paper as 
\begin{equation} \label{eq:linearwk}
\mathbf{w}_k = \begin{cases}
\dfrac{\mathbf{z}_k}{\left\| \mathbf{z}_k \right\|}, & \mbox{MRT},\\
\dfrac{\mathbf{Z}(\mathbf{Z}^H\mathbf{Z})^{-1}\mathbf{e}_k}{\left\| \mathbf{Z}(\mathbf{Z}^H\mathbf{Z})^{-1}\mathbf{e}_k \right\|}, & \mbox{ZF},\\
\dfrac{\mathbf{Z}\left(\mathbf{Z}^H\mathbf{Z} + \sigma^2\mathbf{I}_K\right)^{-1}\mathbf{e}_k}{\left\| \mathbf{Z}\left(\mathbf{Z}^H\mathbf{Z} + \sigma^2\mathbf{I}_K\right)^{-1}\mathbf{e}_k \right\|}, & \mbox{RZF}.
\end{cases}
\end{equation}
where $\mathbf{Z} = [\mathbf{z}_1, \mathbf{z}_2,..., \mathbf{z}_K] \in \mathbb{C}^{M \times K}$ is the channel matrix; $\mathbf{e}_k = [0 \ldots 1 \ldots 0]^T \in \mathbb{C}^K$ is the $k$-th collum of $\mathbf{I}_K$; and $\mathbf{I}_K$ is identity matrix with size $K \times K$. The downlink SER is further derived as in Corollary~\ref{Corollary:Asym}.}
\begin{corollary} \label{Corollary:Asym}
{If the BS deploys the linear beamforming techniques defined in \eqref{eq:linearwk}, the analytical downlink SER of user~$k$ is formulated as
\begin{multline}\label{eq:SERLinear}
        \overline{\mathsf{SER}}_k ( \pmb{\Phi}) =  2\left(1-\frac{1}{\sqrt{m}}\right)\mathrm{erfc}\left(\sqrt{\frac{3\mathsf{SINR}_k ( \pmb{\Phi} )}{m-1}}\right) \\ 
        - \left(1-\frac{1}{\sqrt{m}}\right)^2\mathrm{erfc}^2\left(\sqrt{\frac{3\mathsf{SINR}_k( \pmb{\Phi})}{m-1}}\right),
\end{multline}
where the corresponding SINR of user~$k$ is defined based on the linear beamforming techniques as
\begin{table}[H]
\centering
\begin{tabular}{|c|c|c|}
\hline
 &  $\mathsf{S}_k (\pmb{\Phi})$ & $\mathsf{IN}_k  (\pmb{\Phi}) + \mathsf{NO}_k  (\pmb{\Phi})$ \\  \hline
 MRT &  $\rho \| \mathbf{z}_k\|_2^2$ &  $\displaystyle \sum_{k'=1, k' \neq k}^K \rho \| \mathbf{z}_k^H \mathbf{z}_{k'} \|^2 \|\mathbf{z}_{k'} \|^{-2} + \sigma^2$\\ \hline
ZF& $\rho \left\|\mathbf{Z}^H (\mathbf{Z}^H \mathbf{Z})^{-1}\mathbf{e}_k \right\|^{-2}$ & $\sigma^2$ \\ \hline
 RZF &  $\rho \left| \mathbf{z}_k^H \mathbf{w}_k \right|^2$ & $\displaystyle \sum_{k' = 1, k'\neq k}^K  \rho \left| \mathbf{z}_k^H \mathbf{w}_{k'} \right|^2 +   \sigma^2$\\ \hline
\end{tabular}
\end{table}
\noindent with $\mathsf{SINR}_k(\{ \mathbf{w}_k \}, \pmb{\Phi} ) = \dfrac{\mathsf{S}_k (\pmb{\Phi})}{\mathsf{IN}_k  (\pmb{\Phi})+\mathsf{NO}_k  (\pmb{\Phi})}  $. Here $\mathsf{S}_k (\pmb{\Phi})$ and $\mathsf{IN}_k  (\pmb{\Phi})+\mathsf{NO}_k  (\pmb{\Phi})$ are the signal strength and the mutual interference and noise strength, respectively.
}
\end{corollary}
\begin{proof}
The SINR of user~$k$ is obtained by substituting the linear beamforming vector $\mathbf{w}_k$ in \eqref{eq:linearwk} into \eqref{eq:sinr}. For the MRT technique, we note that $\mathbf{z}_k^H \mathbf{w}_k = \| \mathbf{z}_k \|^2$.  For the ZF technique, we have $|\mathbf{z}_k^H \mathbf{w}_k |^2 = \left|\mathbf{Z}^H (\mathbf{Z}^H \mathbf{Z})^{-1}\mathbf{e}_k \right|^{-2}$ and $|\mathbf{z}_k^H \mathbf{w}_{k'} |^2 =0$,$\forall k \neq k'$ by the mutual interference cancellation.
\end{proof}
{The SINR expression obtained in Corollary~\ref{Corollary:Asym} indicates the benefits of the MRT technique in maximizing the strength of the desired signal. Meanwhile, the ZF technique can effectively cancel out mutual interference, and the RZF balances between those two factors. Apart from these, the analytical downlink SER in \eqref{eq:SERLinear} is only a function of scattering elements, which can be optimized for smart environment control to reduce the  error probability}.
\section{Active and Passive Beamforming Designs for {Symbol Error Rate} Optimization} \label{Sec:Opt}
This section formulates a total SER minimization by optimizing the beamforming vectors and scattering elements. We further propose an improved version of the DE algorithm to solve the optimization problem in polynomial time.
\vspace{-0.25cm}
\subsection{Joint Active and Passive Beamforming Design}
{In each coherence interval, data symbols are simultaneously sent from BS  to the $K$ users, and the SER of each user is computed by \eqref{eq:SER}. Hence, SER is evaluated under a propagation environment including the channels $\mathbf{g}_k$, $\mathbf{u}_k$, and $\mathbf{H}$ together with the controlled scattering elements. In this paper, we formulate the SER optimization problem by treating $\pmb{\Phi}$ and $\{\mathbf{w}_k\}$ as optimization variables, aiming to minimize the average error probability  as follows:
\begin{subequations}\label{ProBlem}
	\begin{alignat}{2}
		& \underset{ \{ \mathbf{w}_k \}, \pmb{\Phi} }{\textrm{minimize}} 
		& & \, \, \mathbb{E}  \left\{\frac{1}{K}\sum_{k=1}^K \mathsf{SER}_k (\{\mathbf{w}_k\}, \pmb{\Phi}) \right\} \label{eq:Obj} \\
		& \textrm{subject to} && \, \,   \rho \|\mathbf{w}_k\|_2^2 \leq P_{\mathrm{max}} \;, \forall k,
        \label{eq:power}\\
        &&& -\pi \leq \theta_n \leq \pi, \forall n,
        \label{eq:element_constraints}
	\end{alignat}
\end{subequations}
where the expectation in \eqref{eq:Obj} is with respect to a given arbitrary signal constellation set.  $P_{\mathrm{max}}$ is the maximum transmit power allocated to data symbols. The objective function of problem \eqref{ProBlem} represents the average SER of each user for multi-user systems. In this context, the constraints \eqref{eq:power}, $\forall k,$ guarantees that the transmit power does not exceed the limit power budget by steering the waveform to user~$k$ with a beamforming technique. Meanwhile, the constraints \eqref{eq:element_constraints}, $\forall n$, determine the feasible set of the phase shift coefficients. In this paper, we focus on the $m$-QAM due to its analytical SER obtained in \eqref{eq:SERLinear}, thus problem~\eqref{ProBlem} becomes
\begin{subequations}\label{ProBlemv1}
	\begin{alignat}{2}
		& \underset{ \{ \mathbf{w}_k \}, \pmb{\Phi} }{\textrm{minimize}} 
		& & \, \, \frac{1}{K}\sum_{k=1}^K \overline{\mathsf{SER}}_k(\{ \mathbf{w}_k \}, \pmb{\Phi})  \\
		& \textrm{subject to} && \, \,   \rho \|\mathbf{w}_k\|_2^2 \leq P_{\mathrm{max}} \;, \forall k,
        \label{eq:power2}\\
        &&& -\pi \leq \theta_n \leq \pi, \forall n.
        \label{eq:element_constraints2}
	\end{alignat}
\end{subequations}
In comparison to \eqref{ProBlem}, problem~\eqref{ProBlemv1} has a lower cost to obtain the solution in each coherence interval, which is more practical since it predetermines the modulation technique.\footnote{We assume that a coherence interval includes a sufficient number of symbols for the SER to apply.} An equivalent formulation of problem~\eqref{ProBlemv1} is in Lemma~\ref{lemma:Equi}.
\begin{lemma} \label{lemma:Equi}
For a given transmit power $0 \leq \rho \leq P_{\max}$, problem~\eqref{ProBlemv1} can be reformulated to 
\begin{subequations}\label{ProBlemv2}
	\begin{alignat}{2}
		& \underset{ \{ \mathbf{w}_k \}, \pmb{\Phi} }{\textrm{minimize}} 
		& & \, \, \frac{1}{K}\sum_{k=1}^K \overline{\mathsf{SER}}_k(\{ \mathbf{w}_k \}, \pmb{\Phi})  \\
		& \textrm{subject to} && \, \, \|\mathbf{w}_k\|_2^2 = \frac{P_{\mathrm{max}}}{\rho} \;, \forall k, \label{eq:power3}\\
        &&& -\pi \leq \theta_n \leq \pi, \forall n,\label{eq:value_phi}
	\end{alignat}
\end{subequations}
which indicates that the beamforming vector of user~$k$ is on a sphere, and therefore narrows the feasible region.
\end{lemma}
\begin{proof}
For given $\{ \mathbf{w}_k \}$ an $\pmb{\Phi}$, let us introduce the SINR of user~$k$ as function of the transmit power based on \eqref{eq:sinr} as
$f_k(\rho) = \rho \left| \mathbf{z}_k^H \mathbf{w}_k \right|^2/(\sum_{k' = 1, k'\neq k}^K  \rho | \mathbf{z}_k^H \mathbf{w}_{k'} |^2 +   \sigma^2)$.
By taking the first-order derivative of $f_k(\rho)$ with respect to $\rho$, one can obtain
$d f_k(\rho) /d\rho = \sigma^2/\left(\sum_{k' = 1, k'\neq k}^K  \rho | \mathbf{z}_k^H \mathbf{w}_{k'} |^2 +   \sigma^2 \right)^2 > 0, \forall k$,
which indicates that $f_k(\rho)$ is monotonically increasing as $\rho$ grows up. In addition, the $\mathrm{erfc}(\cdot)$ function is monotonically descreasing  with $0 \leq \rho \leq P_{\max}$. Consequently, the SER obtained in \eqref{eq:SER} reduces as $\rho$ increases. Alternatively, the communication reliability of all the users will improve by increasing the transmit power. Consequently, at the optimal point, the inequality power constraint holds with equality.
\end{proof}
Even though the beamforming constraints in Lemma~\ref{lemma:Equi} effectively narrow the feasible region, solving problem~\eqref{ProBlemv2} is nontrivial. {Following a similar methodology as in \cite{luo2008dynamic}, problem~\eqref{ProBlemv2} is NP-hard since it can be interpreted as a maximal independent set problem. Nevertheless, the objective function and constraints of problem~\eqref{ProBlemv2} are continuous and bounded. Apart from this, the feasible region is a compact set. According to the Weierstrass' theorem \cite{van2018joint}, the global optimum to problem~\eqref{ProBlemv2} always exists.}}
\vspace{-0.25cm}
\subsection{Passive Beamforming Design with Linear Precoding}
{If the system deploys the linear beamforming techniques in \eqref{eq:linearwk}, we reformulate the average SER minimization to as 
\begin{equation}\label{ProBlemv3}
	\begin{aligned}
		& \underset{ \pmb{\Phi} }{\textrm{minimize}} 
		& & \, \, \frac{1}{K}\sum_{k=1}^K \overline{\mathsf{SER}}_k( \pmb{\Phi})  \\
		& \textrm{subject to} && -\pi \leq \theta_n \leq \pi, \forall n,
	\end{aligned}
\end{equation}
where the analytical downlink SER of user~$k$ is given in \eqref{eq:SERLinear}. We stress that jointly optimizing both the active and passive beamforming as in \eqref{ProBlemv2} needs $MK$ complex variables and  $N$ real phase shift variables with the $K+N$  constraints.  Problem~\eqref{ProBlemv3} is a suboptimal design compared to that of \eqref{ProBlemv2}, but it only contains $N$ variables and $N$ constraints. The inherent non-convexity still remains in \eqref{ProBlemv3} even though the optimal solution exists.}
\begin{remark}
Since SER is a function of the SINR as obtained in \eqref{eq:SER}, a straightforward solution to improve the communication reliability of a particular user can be obtained by increasing the SINR. Nonetheless, it may {increase the SER}
of the other users due to mutual interference. Hence, minimizing the total SER of the entire network as in \eqref{ProBlemv2} is a challenging task.
\end{remark}


\begin{figure}[t]
    \centering
    \includegraphics[width = 0.5\textwidth]{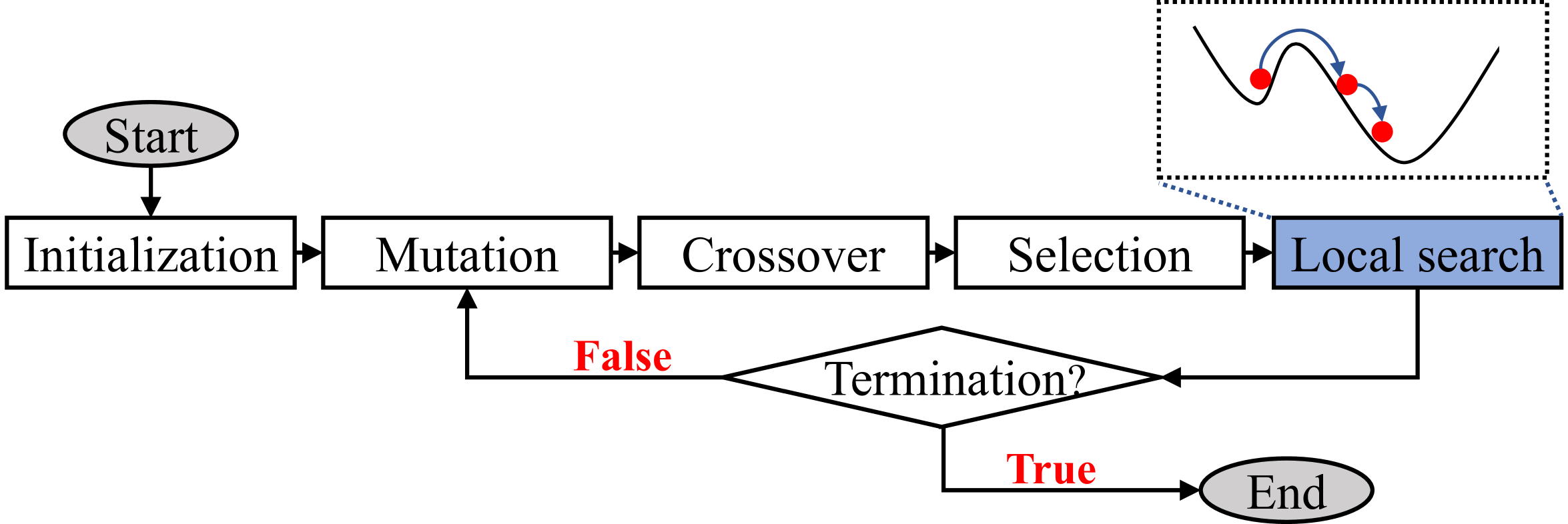}
    \caption{Followchart of the joint active and passive beamforming designs by the improved DE algorithm with a local search.}
    \label{fig:flowchart}
    \vspace{-0.5cm}
\end{figure}
\vspace{-0.3cm}
\section{Improved Differential Evolution-based Active and Passive Beamforming Designs} \label{Sec:DE}
{This section provides an efficient algorithm to solve problem~\eqref{ProBlemv2} in polynomial time based on an improved version of the DE. As a consequence, one can also exploit the proposed algorithm to obtain the phase shift design in problem~\eqref{ProBlemv3}.}

\subsection{Improved Different Evolution with Local Search}

{We propose an improved version of the DE method to design the phase shifts and the beamforming vectors by solving problem~\eqref{ProBlemv2} with the flowchart presented in Fig.~\ref{fig:flowchart}. The proposed algorithm begins with establishing an initial population, subsequently overseeing and enhancing this population across generations through evolution. Within each generation, the mutation and crossover mechanisms are applied to each individual to generate new individuals. Each new individual undergoes evaluation in comparison to its parent: If the offspring demonstrates better fitness, i.e., the objective function of problem~\eqref{ProBlemv2}, it replaces the parent in the next generation. At the end of each generation, the individual with the minimum fitness value is designated as the current solution to  problem~\eqref{ProBlemv2}. To mitigate the risk that the individuals are entrenched in a local minimum, we exploit a local search, guided by the improvements in fitness, to facilitate an escape from the local minimum. Moreover, the efficacy of the improved DE algorithm hinges on two pivotal parameters consisting of the scale parameter $\mathsf{F}$ and the crossover rate parameter $\mathsf{CR}$. In contrast to the conventional approach of predefining these parameters used in the standard DE, i.e., see \cite{shade} and references therein, we dynamically adapt them based on evolving search behaviors. Our proposal is described in Algorithm~\ref{alg:Improved_shade}.}

\begin{algorithm}[t]
    \caption{Improved DE Beamforming Design}
    \label{alg:Improved_shade}
        \textbf{Input:} {The channels $\mathbf{H}$, $\left\{\mathbf{g}_k\right\}_{k = 1}^K$, and $\left\{\mathbf{u}_k\right\}_{k = 1}^K$; the initial scale factor $\mathsf{F}_{\mathrm{init}}$; the initial crossover rate $\mathsf{CR}_{\mathrm{init}}$; the maximum number of generations $G_{\max}$; the maximum number of fitness evaluations $\mathsf{NFE}_{\max}$ and local search parameter $\tilde{\sigma}$}\\
        \textbf{Output:} The phase shift matrix $\pmb{\Phi}$ and the beamforming vectors $\{\mathbf{w}_k\}$
    \begin{algorithmic}[1] 
        \STATE The generation index $G \leftarrow 1$;
        \STATE Randomly initialize a population $\mathcal{P}^{(1)}$ of $I$ individuals.
        \STATE $\mathsf{MCR}_i \leftarrow \mathsf{CR}_{\operatorname{init}}$;
         $\mathsf{MF}_i \leftarrow \mathsf{F}_{\operatorname{init}}$ $\forall i \in \{1,\ldots,H\}$
        \WHILE{\textit{Termination condition not met}}
%
            \STATE  $\mathcal{B}^{(G)} \leftarrow \varnothing$
            \FOR{$i\leftarrow 0$ to $I$}
                \STATE Calculate $\mathsf{CR}^{(iG)}$ and $\mathsf{F}^{(iG)}$ as in \eqref{eq:CRMCR};
                \STATE Generate mutant vector $\mathbf{u}^{(iG)}$ using \eqref{eq:mutation};
                \STATE Generate trial solution $\pmb{\omega}^{(iG)}$  
                using \eqref{eq:crossover};
                \IF{$f\left(\pmb{\omega}^{(iG)}\right) \leq f\left(\mathbf{x}^{(iG)}\right)$}
                    \STATE $\mathcal{B}^{(G)} \leftarrow \mathcal{B}^{(G)} \cup \pmb{\omega}^{(iG)} $;
                \ELSE
                    \STATE $\mathcal{B}^{(G)} \leftarrow \mathcal{B}^{(G)} \cup \mathbf{x}^{(iG)} $
                \ENDIF
            \ENDFOR
            \STATE Update $\mathsf{MCR}$ and $\mathsf{MF}$ using method \cite{shade};
        \STATE Create population $\mathcal{Q}^{(G)}$ based on $\mathcal{B}^{(G)}$ by using local search step;
            \STATE $G \leftarrow G + 1$
            \STATE $\mathcal{P}^{(G)} \leftarrow \mathcal{Q}^{(G)}$
        \ENDWHILE
    \end{algorithmic}
\end{algorithm}
\subsubsection{Solution representation and population initialization}
{Each individual represents a potential solution, and the fitness value of an individual corresponds to the objective function of problem~\eqref{ProBlemv2}. We represent each individual by an $(N + 2MK)$-dimensional vector, where each dimension has real values within the range $[-1, 1]$. This vector can be adapted into a phase shift matrix $\pmb{\Phi}$ and the beamforming vectors $\{\mathbf{w}_k\}$.\footnote{The encoded individual indicates that a complex beamforming vector $\mathbf{w}_k \in \mathbb{C}^M$ is rearranged into a real one of length $2M$.}
Specifically, the first $N$ elements correspond to the phase shift matrix $\pmb{\Phi}$. To convert them into the phase shift matrix, we multiply each element by scale parameter $\pi$ to ensure the constraints \eqref{eq:value_phi}. The remaining $2MK$ elements are divided into the $K$ segments, corresponding to the $K$ users. Note that two consecutive elements should represent a complex number corresponding to one element of the beamforming vector $\mathbf{w}_k$. After the conversion from segments to the beamforming vector $\mathbf{w}_k$, we normalize $\mathbf{w}_k, \forall k,$ to satisfy the constraints \eqref{eq:power3}.}
\begin{example}
{Let us consider a toy example with the configuration $N=4$, $M=2$, and $K=2$ having an individual illustrated in Fig.~\ref{fig:individual}. After decoding the individual, we obtain the phase shift matrix $\pmb{\Phi}$ as follows 
\begin{equation}
\pmb{\Phi} = \mathrm{diag}\left(\big[ e^{0.3\pi j}, e^{-0.5\pi j}, e^{-0.2\pi j}, e^{0.7\pi j} \big]^T\right),
\end{equation}
and the intermediate version of the beamforming vectors are 
\begin{align}
&\tilde{\mathbf{w}}_1 = [-0.3 + 0.4j, 0.5 - 0.5j]^T,\\
&\tilde{\mathbf{w}}_2 = [0.6 + 0.8j, -0.4 + 0.3j ]^T.
\end{align}
By performing the normalization and then scaling up by the factor $\sqrt{P_{\max}/\rho}$, the beamforming vectors $\mathbf{w}_1$ and $\mathbf{w}_2$ are respectively obtained as follows
\begin{align}
    &\mathbf{w}_1 = \frac{\tilde{\mathbf{w}_1}\sqrt{P_{\max}}}{\|\tilde{\mathbf{w}}_1\|\sqrt{\rho}} = \sqrt{\frac{P_{\max}}{\rho}}\left[-\frac{\sqrt{3}}{5} + \frac{4\sqrt{3}}{15}j, \frac{\sqrt{3}}{3} - \frac{\sqrt{3}}{3}j\right]^T,\\
    &\mathbf{w}_2 = \frac{\tilde{\mathbf{w}}_2\sqrt{P_{\max}}}{\|\tilde{\mathbf{w}}_2\|\sqrt{\rho}} = \sqrt{\frac{P_{\max}}{\rho}} \left[\frac{6\sqrt{5}}{25} + \frac{8\sqrt{5}}{25}j, -\frac{4\sqrt{5}}{25} + \frac{3\sqrt{5}}{25}j\right]^T,
\end{align}
which satisfy the constraints \eqref{eq:power3} of problem~\eqref{ProBlemv2}.}
\end{example}
\begin{figure}[t]
    \centering
    \includegraphics[width =2.8in, trim = 1cm 0.2cm 0cm 1.78cm, clip = true]{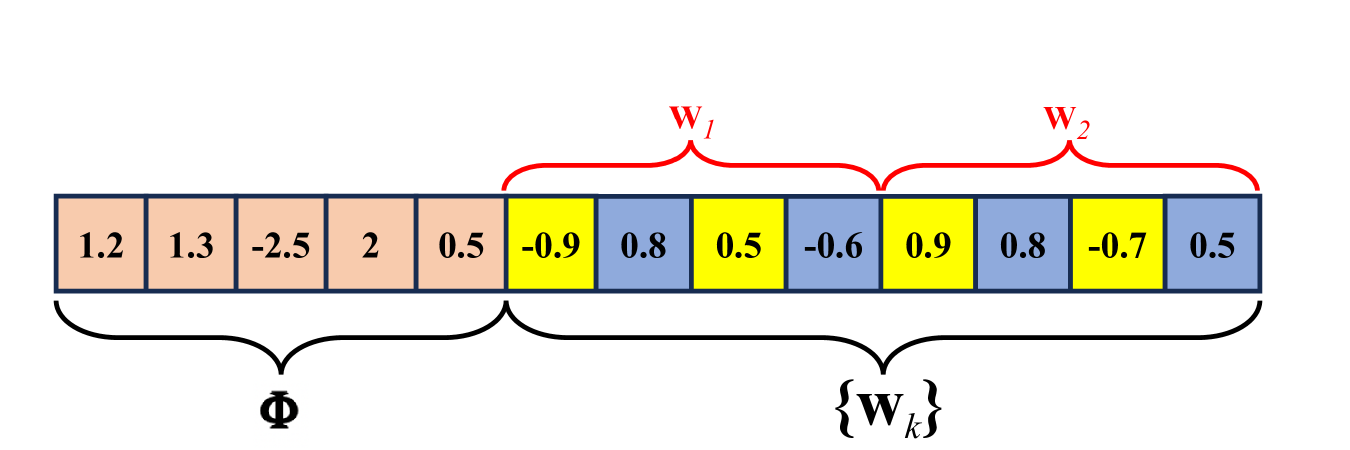}
    \caption{An individual with $N = 4$, $M = 2$, and $K = 2$.}
    \label{fig:individual}
    \vspace{-0.5cm}
\end{figure}
{At the initial stage, a population $\mathcal{P}^{(1)}$ is initialized that consists of $I$ individuals, i.e., $|\mathcal{P}^{(1)} | = I$, each is randomly initialized. To identify individuals within the population across generations, we designate the $i$-th individual of the population in the $G$-th generation, denoted by $\mathbf{x}^{(iG)} \in \mathbb{R}^{N+ 2MK}$ as:
\begin{equation}
    \mathbf{x}^{(iG)} =  \left[x_1^{(iG)},x_2^{(iG)}, \ldots, x_N^{(i G)},x_{N+1}^{(iG)},\ldots, x_{N + 2MK}^{(iG)}  \right]^T,
\end{equation}
where $x_n^{(iG)} \in  [-1, 1]$ if $n \in \{1,\ldots, N\}$, which represents the $n$-th phase shift coefficient. Otherwise, $x_n^{(iG)} \in  [-1, 1]$ with $n \in \{N+1,\ldots 2MK \}$ model either the real or imaginary of the beamforming coefficients.
Additionally, we denote $\mathbf{x}^{(\operatorname{best} G)}$ as the individual with the minimum fitness in the $G$-th generation that is obtained by evaluating the objective function of problem~\eqref{ProBlemv2} for the current population.
\subsubsection{Mutation}
In the $G$-th generation, each individual $\mathbf{x}^{(iG)}$ will create a mutant vector, denoted by $\mathbf{u}^{(iG)} \in \mathbb{R}^{N + 2MK}$ as
\begin{equation}
\mathbf{u}^{(iG)}=\left[ u_1^{(iG)},u_2^{(iG)}, \ldots,  u_N^{(iG)}, u_{N+1}^{(iG)}, \ldots, u_{N + 2MK}^{(iG)} \right]^T,
\end{equation}
through the mutation operators. In this paper, we first use the $DE/best/1$ operator to create a mutant vector, which is mathematically described as 
\begin{equation}\label{eq:mutation}
    \mathbf{u}^{(iG)} = \mathbf{x}^{(\operatorname{best} G)} + \mathsf{F}^{(iG)} \left(\mathbf{x}^{(r_1G)} - \mathbf{x}^{(r_2G)}\right),
\end{equation}
where $\mathsf{F}^{(iG)}$ denotes the scale factor parameter assigned to the $i$-th individual to augment the differential vectors. Specifically,  $\mathbf{x}^{(r_1G)} - \mathbf{x}^{(r_2 G)}$ represents the disparity between the two randomly selected individuals, i.e., $\mathbf{x}^{(r_1G)}$ and $\mathbf{x}^{(r_2 G)}$, within the population. It is noteworthy that the $i$-th individual should possess its specific scale factor $\mathsf{F}^{(iG)}$ contingent on its behavior, as detailed in \cite{shade}. 
Nonetheless, \eqref{eq:mutation} may generate some elements of $\mathbf{u}^{(iG)}$ that violate the value domain constraint. To handle this issue, we employ the normalization procedures to ensure that the elements of each individual are in the range $[-1,1]$. Concretely, the first $N$ elements are normalized as 
\begin{equation}
    {u}^{(iG)}_n \leftarrow {u}^{(iG)}_n + 2 \left\lfloor (1-{u}^{(iG)}_n)/2 \right\rfloor  , \ \forall n \in \{1, \ldots, N\}.
    \label{eq:normal_angle}
\end{equation}
The normalization in \eqref{eq:normal_angle} ensures that the phase shift coefficients of the matrix $\pmb{\Phi}$ lie on the unit circle and the trigonometric properties remain unchanged after normalization.
The remaining part of $\mathbf{u}^{(iG)}$ regarding the remaining $2MK$ elements of the beamforming vectors $\{\mathbf{w}_k\}$ are normalized by 
\begin{equation}
    {u}^{(iG)}_n  = \begin{cases}
        (- 1 + {x}^{(iG)}_n)/2, & \mbox{if }  {u}^{(iG)}_n < -1,\\
        (1 + {x}^{(iG)}_n)/2, & \mbox{if }  {u}^{(iG)}_n > 1,
    \end{cases}  
    \label{eq:normal_w}
\end{equation}
with $n \in \{N+1, \ldots, N+ 2MK\}$.  Although the individuals comprise both the passive and active beamforming coefficients with different properties, the procedures in \eqref{eq:normal_angle} and \eqref{eq:normal_w} keep a homogeneous population without dominant elements.}

\subsubsection{Crossover}
{After the mutation, a trial vector $\pmb{\omega}^{(iG)} \in \mathbb{R}^{N + 2MK}$ is generated from a combination of $\mathbf{x}^{(iG)}$ and $\mathbf{u}^{(iG)}$ via the crossover operator where the $n$-th element, denoted by $\omega^{(iG)}_n$, is defined as
\begin{equation}
    \omega^{(iG)}_n = \begin{cases}
        {u}^{(iG)}_n, & \mbox{if }  \mathcal{U}[0,1] \leq  ~\mathsf{CR}^{(iG)}~ \mbox{ or } n = n_{\operatorname{rand}}, \\
        {x}^{(iG)}_n ,& \mbox{otherwise},
    \end{cases}
    \label{eq:crossover}
\end{equation}
where $\mathcal{U}[0,1]$ generates a random variable uniformly distributed in the range $[0,1]$. The parameter $\mathsf{CR}^{(iG)}$ is flexibly selected to adapt each individual $\mathrm{\omega}^{(iG)}$. We note that the condition $\mathcal{U}[0,1] \leq  ~\mathsf{CR}^{(iG)}$ may lead to a situation in which the descendants are totally the same as their parents. Consequently, $n_{\mathrm{rand}} \in \{1, \ldots, N+2MK\}$ is randomly preselected to ensure that $\pmb{\omega}^{(iG)}$ has at least one dimension different from $\mathbf{x}^{(iG)}$, indicating that after the evolutionary process, the offspring is not identical to the parent.}

\subsubsection{Selection}
{After each trial vector $ \pmb{\omega}^{(iG)}$ is generated, the corresponding phase shift coefficients $\pmb{\Phi}^{(iG)}$ and beamforming vectors $\{ \mathbf{w}_k^{(iG)} \}$ are obtained. These potential solutions will be evaluated by using the fitness function that is defined from the objective function of problem \eqref{ProBlemv2} as
\begin{equation}
\label{eq:fitness_function}
f (\{ \mathbf{w}_k^{(iG)} \}, \pmb{\Phi}^{(iG)} ) =  \frac{1}{K}\sum_{k=1}^K \overline{\mathsf{SER}}_k(\{ \mathbf{w}_k^{(iG)}  \}, \pmb{\Phi}^{(iG)} ),
\end{equation}
where the SER of user~$k$ is computed as in \eqref{eq:SER}. If the fitness value obtained by the pair $\{\{ \mathbf{w}_k^{(iG)} \}, \pmb{\Phi}^{(iG)}\}$  is better than that of individual $\mathbf{x}^{(iG)}$, $\pmb{\omega}^{(iG)}$ will replace $\mathbf{x}^{(iG)}$ in the population. Alternatively, $\mathbf{x}^{(iG)}$ is removed from the population. The selection process always ensures the non-increasing objective function of problem~\eqref{ProBlemv2} across the generations.}

\subsubsection{Local search}
{Even though the DE algorithm can find the solution in a short period of time, its drawback is that the search process may get trapped in a local optima. When some individuals of the population are located in local optima, the evolution process can lead all the remaining to move toward this locality, causing premature convergence. In this paper, to prevent solutions from prematurely converging, we propose to exploit a local search technique to refine some of the solutions from DE, enabling them to escape local optima as illustrated in Fig.~\ref{fig:flowchart}. 
After the selection, we randomly choose $\lambda_G \geq 0$ individuals for consideration to move them to nearby positions. If the new position is better than the previous one, the corresponding new individual replaces the old one in the population. To ensure that the local search process does not increase computational time drastically, as it becomes less likely to find better potential positions once approaching the stable solution, the number of selected individuals is 
\begin{equation}
    \lambda_G = \left\lfloor \frac{f^{(cG)} (G_{\max} - G) I }{f^{(\operatorname{best}(G-1))}G_{\max}} \right\rfloor,
    \label{eq:parameter_local_search}
\end{equation}
where $f^{(cG)}$ represents the best fitness value after the selection in the $G$-th generation, which is
\begin{equation}
f^{(cG)} = \min_{i}\{f (\{ \mathbf{w}_k^{(iG)} \}, \pmb{\Phi}^{(iG)} ) \},
\end{equation}
while $f^{(\mathrm{best}(G-1))}$ is the best fitness value in the $(G-1)$-th generation. In \eqref{eq:parameter_local_search}, the fraction $(G_{\max} -G)I/G_{\max}$ implies that the maximum number of local search operations will decrease gradually over generations. Along the generations, the number of conducted local search operations is inversely proportional to the improvement achieved by the improved DE indicating that Algorithm~\ref{alg:Improved_shade} is effective. To move the solution to a neighbor location, we create a Gaussian vector $\pmb{\xi}^{(iG)} \sim \mathcal{N}(\mathbf{0}, \tilde{\sigma}^2 \mathbf{I}_{N + 2MK})$ with the small variance $\tilde{\sigma}^2$ effective in exploring the neighborhoods. For an individual $\pmb{\omega}^{(iG)}$ in $\mathcal{B}^{(G)}$, the neighboring individual is generated as
\begin{equation}
    \tilde{\pmb{\omega}}^{(iG)} = \pmb{\omega}^{(iG)} + \pmb{\xi}^{(iG)}.
    \label{eq:neighbor_solution}
\end{equation}
In this case, $\tilde{\pmb{\omega}}^{(iG)}$ is a potential solution and its fitness value is computed accordingly. If the obtained fitness value is better than that of $\pmb{\omega}^{(iG)}$, then $\tilde{\pmb{\omega}}^{(iG)}$ is chosen for $\mathcal{Q}^{(G)}$, conversely $\pmb{\omega}^{(iG)}$ is chosen.}
\subsubsection{Parameter adaptation}
{The performance of Algorithm~\ref{alg:Improved_shade} is impacted by several hyperparameters such as $\mathsf{F}^{(iG)}$ in \eqref{eq:mutation} and $\mathsf{CR}^{(iG)}$ in \eqref{eq:crossover} because of their crucial roles in generating new solutions. Instead of keeping these parameters fixed as is done in the canonical DE, we have incorporated a technique known as success-history-based parameter adaptation (SHADE) \cite{shade} into our proposed algorithm to enable an automatic adjustment of these hyperparameters. For the mutation of the $i$-th individual, we introduce two memory arrays of length $H \geq 1$, denoted as $\mathsf{MCR}$ and $\mathsf{MF}$, to retain information regarding the successful crossover rate and scale factors. These values, which have contributed in producing improved the solutions compared to the previous generations, are stored and later utilized to guide the algorithm in determining the crossover rates and scale factors. In the $G$-th generation, the parameters $\mathsf{F}^{(iG)}$ and $\mathsf{CR}^{(iG)}$  will be updated by the behaviors of the corresponding $i$-th individual as
\begin{equation} \label{eq:CRMCR}
    \mathsf{CR}^{(iG)} \sim \mathcal{N}\left(\mathsf{MCR}_{r_i}, 0.1\right) \mbox{ and } \mathsf{F}^{(iG)} \sim \mathcal{N}\left(\mathsf{MF}_{r_i}, 0.1\right),
\end{equation}
where $r_i$ is an index randomly selected in the range of $[1,H]$. If $\mathsf{CR}^{(iG)} , \mathsf{F}^{(iG)} \notin [0, 1]$,  these parameters are assigned the initial values, which are $\mathsf{CR}_{\mathrm{init}}$ and $\mathsf{F}_{\mathrm{init}}$.}
\subsubsection{Termination condition}
{Algorithm~\ref{alg:Improved_shade} will be terminated if one of the following three criteria is satisfied: \textit{i)} the number of generations reaches a maximum specified, \textit{ii)} the number of evaluations for the fitness reaches a maximum specified, and \textit{iii)} the objective function remains unchanged over predetermined generations.}
\vspace{-0.25cm}
\subsection{Convergence and Computational Complexity}
{We now probabilistically analyze the convergence of Algorithm~\ref{alg:Improved_shade} to an optimal solution. Over the generations, solutions to the problem gradually converge towards an optimal value or a local optimum. The convergence of Algorithm~\ref{alg:Improved_shade} is determined based on the following theorem.}
\begin{theorem}\label{theorem:local_convergence}
{Let us introduce  $\mathcal{S}_{\varepsilon}^{*}$ to be the space of the $\varepsilon$-optimal solution to ~\eqref{ProBlemv2}, which is
\begin{multline} \label{eq:Sep}
\mathcal{S}_{\varepsilon}^{\ast} = \Big\{ \{\mathbf{w}_k \}, \pmb{\Phi} \, \big| \,  \lvert f(\{\mathbf{w}_k \}, \pmb{\Phi} ) - f(\{\mathbf{w}_k^{\ast} \}, \pmb{\Phi}^\ast) \rvert < \varepsilon, \\ \{\mathbf{w}_k \}, \pmb{\Phi} \in \mathcal{S} \mbox{ and } \{\mathbf{w}_k^\ast \}, \pmb{\Phi}^\ast \in \mathcal{S}\Big\},
\end{multline}
where $f(\cdot)$ and  $\mathcal{S}$ denote the objective function and the feasible region of in problem \eqref{ProBlemv2}, respectively. Besides,   $\{\mathbf{w}_k^\ast \}, \pmb{\Phi}^\ast$ is the optimal solution and $\varepsilon$ is a small positive value. For a population $\mathcal{P}$ with $I$ individuals, the probability of the population converging to an individual belonging to $\mathcal{S}_{\varepsilon}^{*}$ by exploiting Algorithm~\ref{alg:Improved_shade} is defined as follows
\begin{equation} \label{eq:LowerBound}
    \begin{split}
         &\mathsf{Pr}\left(\mathcal{P} \cap \mathcal{S}_{\varepsilon}^{\ast} \neq \varnothing \right) \geq  1 -\left(1 - \mu_1\left(\mathcal{S}_{\varepsilon}^{\ast}\right)P_{\mathrm{ep}} \right )^{I}  \times \\
          &\left(1 - \mu_2\left(\mathcal{S}_{\varepsilon}^{\ast}\right)\left(\frac{1}{\tilde{\sigma}\sqrt{2\pi}}\right)^{N+2MK}e^{\frac{-2(N + 2MK)}{\tilde{\sigma}^2}}\right)^{\tilde{\lambda}},      
    \end{split}
\end{equation}
where $P_{\mathrm{ep}} \in [0, 1]$ is the mutation probability of each individual; $\mu_1 \left(\mathcal{S}_{\varepsilon}^{*}\right)$ and $\mu_2 \left(\mathcal{S}_{\varepsilon}^{*}\right)$ are the measures to the space $\mathcal{S}_{\varepsilon}^{\ast}$ regarding the mutation and the local search, respectively; and $\tilde{\lambda}$ as the number of individuals utilized for the local search step in each generation.}  
\end{theorem}
\begin{proof}
{The proof is based on computing the probability that the optimal solution appears in the population as the generations grow. The detailed proof is available in  Appendix~\ref{appen:local_convergence}.}
\end{proof}
Theorem \ref{theorem:local_convergence} unveils that Algorithm~\ref{alg:Improved_shade} with the local search has better convergence probability than the standard DE. In more detail, the following inequality holds
\begin{equation}
\begin{split}
 & \left(1 - \mu_1\left(\mathcal{S}_{\varepsilon}^{\ast}\right)P_{\mathrm{ep}} \right )^{I} \left(1 - \mu_2\left(\mathcal{S}_{\varepsilon}^{\ast}\right)\left(\frac{1}{\tilde{\sigma}\sqrt{2\pi}}\right)^{N+2MK}e^{\frac{-2(N + 2MK)}{\tilde{\sigma}^2}}\right)^{\tilde{\lambda}} \\
 & \leq \left(1 - \mu_1\left(\mathcal{S}_{\varepsilon}^{\ast}\right)P_{\mathrm{ep}} \right )^{I},
\end{split}
\end{equation}
whose right-hand side is obtained by a canonical form of the DE. It demonstrates that the local search improves the lower bound of the convergence probability as shown in \eqref{eq:LowerBound}. {Furthermore, Theorem~\ref{theorem:local_convergence} also indicates that Algorithm~\ref{alg:Improved_shade} does converge to the global optimum from an initial population within the feasible domain after a sufficiently large number of the generations as shown in Corollary~\ref{corollary:converge}.
\begin{corollary}
\label{corollary:converge}
Let $\mathcal{P}^{(G)}$ be the population at the $G$-th generation, then the convergence to the global optimum of problem~\eqref{ProBlemv2} in the probability obtained by Algorithm~\ref{alg:Improved_shade} is
    \begin{equation}
    \lim_{G\rightarrow \infty}\  \mathsf{Pr}\left(\mathcal{P}^{(G)}\cap \mathcal{S}^*_{\varepsilon} \neq \varnothing \right) = 1.
\end{equation}
\end{corollary}
\begin{proof}
The proof is accomplished as a consequence of Theorem~\ref{theorem:local_convergence} and the sandwich theorem. The detailed proof is available in Appendix~\ref{Appdixcorollary:converge}.
\end{proof}
}
We now analyze the computational complexity of Algorithm~\ref{alg:Improved_shade} based on the basic operations with the complexity of $\mathcal{O}(1)$ including addition, subtraction, multiplication, inversion, composition, and transpose. For user~$k$, evaluating the aggregated channel $\mathbf{z}_k = \mathbf{u}_k + \mathbf{H}\pmb{\Phi}\mathbf{g}_k$ requires a computational complexity in the order of  $\mathcal{O}\left(MN^2\right)$.  Conditioned on $\mathbf{z}_k, \forall k$, the computational complexity of evaluating 
$\mathsf{SINR}_k(\{\mathbf{w}_k\}, \Phi)$  is in the order of $\mathcal{O}(KM)$. The complexity to compute the SER for user~$k$ is thus in the order of $\mathcal{O}\left(MN^2 + MK\right)$, it leads to that of the objective function is 
$\mathcal{O}\left(KMN^2 + MK^2\right)$. Regarding the computational complexity of Algorithm \ref{alg:Improved_shade}, the initialization step requires $\mathcal{O} (IN + 2IKM)$. In each generation, sorting the population to extract the best individual for \eqref{eq:mutation} requires $\mathcal{O}(I\log(I))$. The mutation and crossover both require $\mathcal{O}\left(IN + 2IKM\right)$. The selection  requires $\mathcal{O}(I)$. The local search  requires $\mathcal{O}\left(\tilde{\lambda}N + 2\tilde{\lambda}KM\right)$, while the parameter adaption needs $\mathcal{O}(I)$. Therefore, the  computatinal complexity of Algorithm~\ref{alg:Improved_shade} is $\mathcal{O}\left(GI\log(I) + 2GI(N + 2KM))\right)$, where $G$ is the number of generations.
\begin{remark}
Algorithm~\ref{alg:Improved_shade} provides an efficient approach to obtain the smart phase-shift control that minimizes the total SER of multiple-access networks. This meta-heuristic algorithm with a local search converges to the global optimum from an initial feasible point with better probability than the pure DE. Different from \cite{ye2020joint, luo2021spatial}, our proposed framework is not based on the first and second derivatives of the objective, which gives a huge advantage in the reduction of the computational complexity and is easily implemented in practice since the SER includes the indefinite integrals.
\end{remark}

\vspace{-0.25cm}
\section{Numerical Results} \label{Sec:NumRe}
{This section provides experimental results to verify {the SER minimization quality}
based on \eqref{eq:SER} and the effectiveness of the active and passive beamforming design utilizing Algorithm~\ref{alg:Improved_shade}.  
In particular,  we consider a system where the BS is equipped with $100$ antennas, and the RIS consists of $256$ phase shift elements. The number of users is $K \in \{2, 10, 50 \}$, up to the network density. The signal modulation and demodulation are the $16$-QAM. The propagation channels $\mathbf{H}$, $\mathbf{u}_k$, and $\mathbf{g}_k$ are constructed as
$\mathbf{H} = \sqrt{\beta_{\mathrm{bsr}}} \tilde{\mathbf{H}}, \mathbf{g}_k = \sqrt{\beta_{rk}} \tilde{\mathbf{g}}_k, \mathbf{u}_k = \sqrt{\beta_{k}} \tilde{\mathbf{u}}_k,$
where $\beta_{sr}$, $\beta_{rk}$, and $\beta_{k}$ are the large-scale fading coefficients that express the propagation distance between the BS and the RIS, the RIS and user~$k$, and the BS and user~$k$, respectively. Meanwhile, $\tilde{\mathbf{H}}$, $\tilde{\mathbf{g}}_k$, and $\tilde{\mathbf{u}}_k$ follow the Rayleigh fading, which is a quasistatic block fading model, where the channel remains fixed for the considered binary sequence.} 

\begin{figure*}[t]
    \centering
    \subfloat[$K=2$ users]{\includegraphics[width =0.3\linewidth]{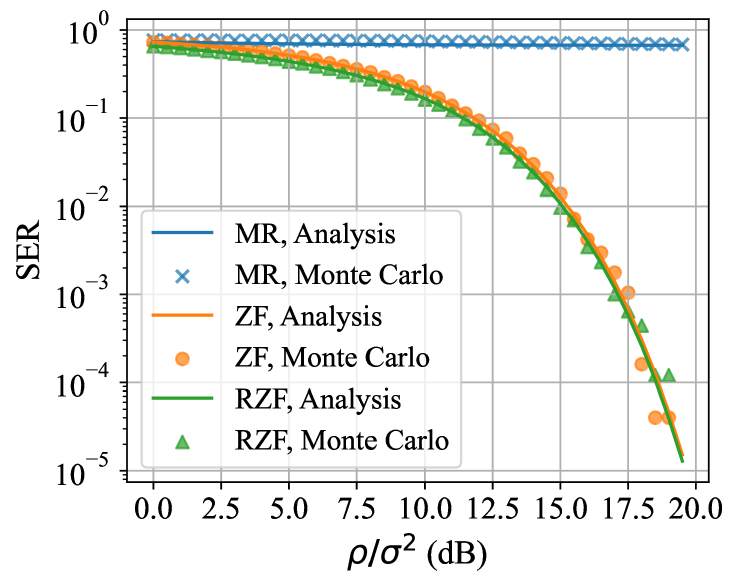}}
    \subfloat[$K=10$ users]{\includegraphics[width =0.3\linewidth]{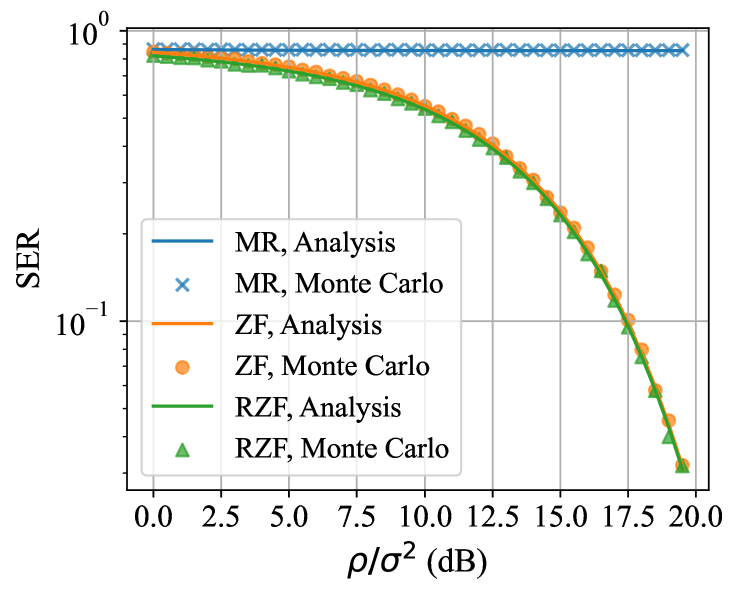}}
    \subfloat[$K=50$ users]{\includegraphics[width = 0.3\linewidth]{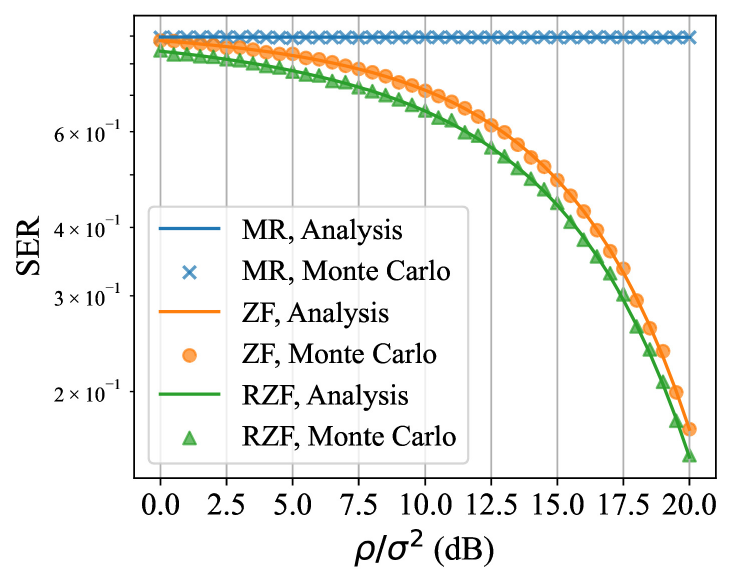}}
    \caption{SER of a RIS-assisted multiuser MIMO system with the different number of users in the network.}
    \label{fig:result_simulation}
    \vspace{-0.5cm}
\end{figure*}

\begin{figure*}[t]
    \centering
    \subfloat{\includegraphics[width=0.3\textwidth]{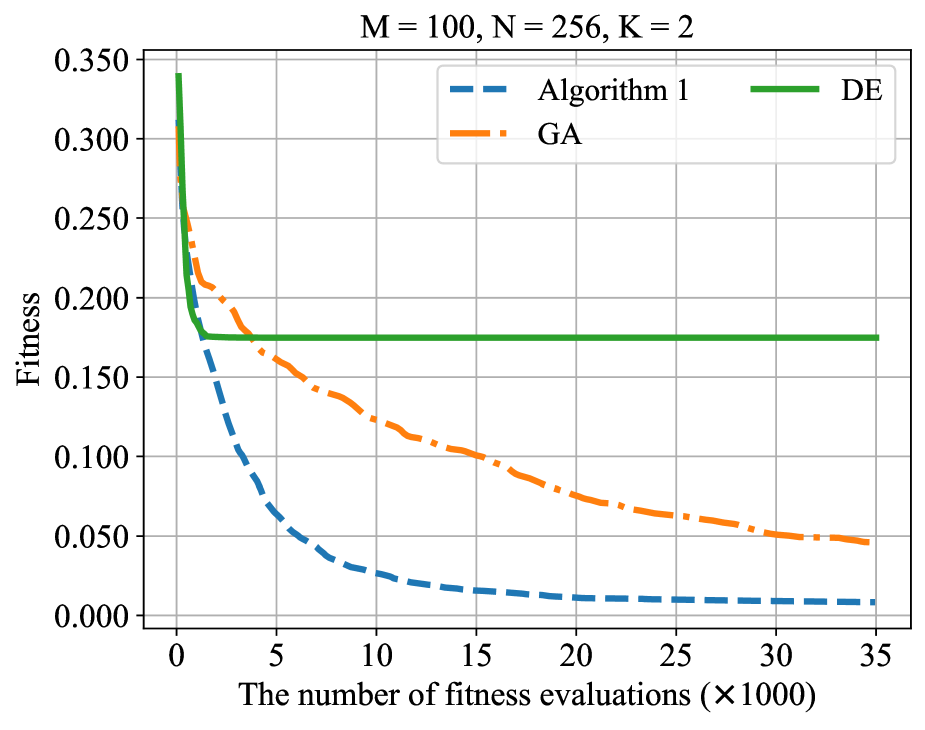}}
    \subfloat{\includegraphics[width=0.3\textwidth]{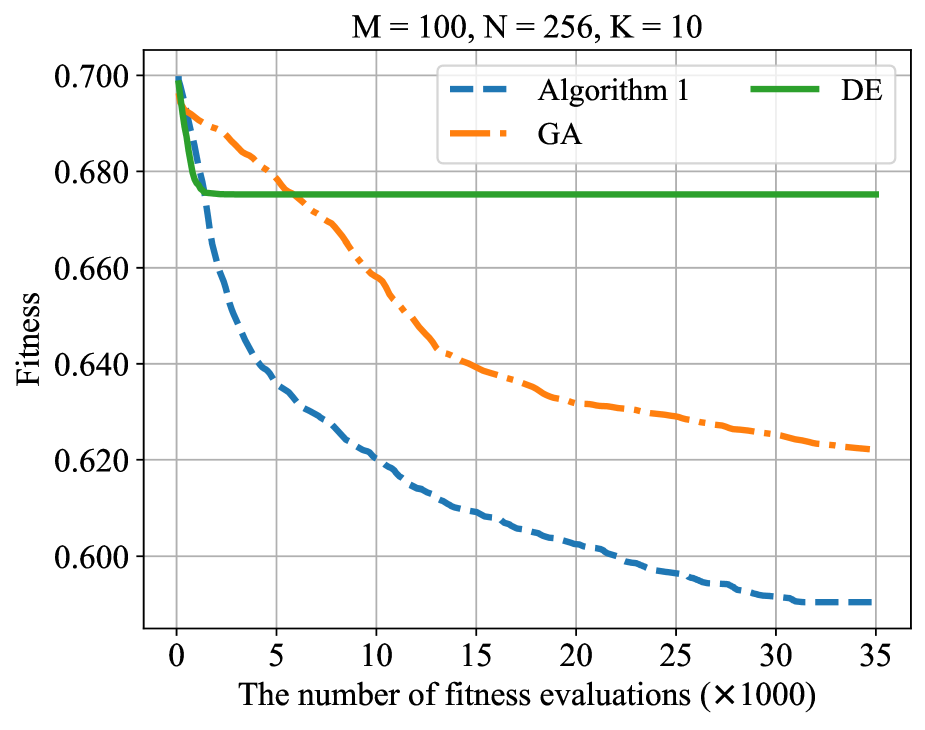}}
    \subfloat{\includegraphics[width=0.3\textwidth]{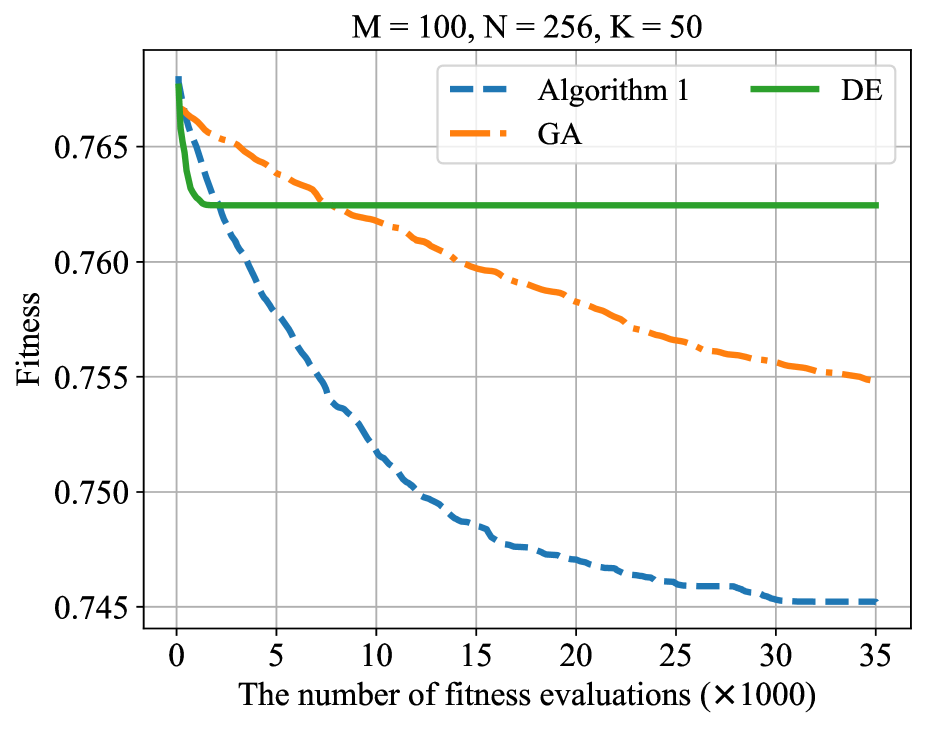}}
    \caption{Convergence of different benchmarks comprising Algorithm~\ref{alg:Improved_shade}, the GA in \cite{peng2021analysis}, and the DE in \cite{huang2022placement} with $\rho/\sigma^2 =5$ [dB].}
    \label{fig:gen_convergence}
    \vspace{-0.5cm}
\end{figure*}
{In Fig.~\ref{fig:result_simulation}, we validate the SER obtained in Theorem~\ref{theorem:ser} and Monte Carlo simulations by exploiting \eqref{eq:SERk} with the linear precoding techniques in \eqref{eq:linearwk}. Each user sends a binary data sequence having the length of $100000$ bits. The results unveil the close alignment between the obtained analytical SER and Monte-Carlo simulations that demonstrates the accuracy of the treatment for mutual interference in \eqref{eq:INkNOk}. The phase shifts are randomly initialized in the range $[-\pi, \pi]$. Furthermore, it is evident that the SER of each user varies depending on the precoding technique. With the same phase shift controlled coefficients $\pmb{\Phi}$, the MR technique performs less effectively compared to the other methods due to mutual interference-dominated scenarios. As the system serves a few users, the ZF and RZF techniques exhibit similar effectiveness. If the number of users increases making mutual interference severe, the RZF technique demonstrates noticeable effectiveness. Fig.~\ref{fig:result_simulation}(c) that the system can simultaneously serve $50$ users at the same time and frequency resource. It means that a huge amount of users can be served at the end of the day. The SER is high compared to the other settings because of strong mutual interference. A better SER should be obtained by the channel coding methods \cite{phan2023deep}, which should be a potential research direction.}

{In order to demonstrate the benefits of Algorithm~\ref{alg:Improved_shade} with an additional local search, Fig.~\ref{fig:gen_convergence} visualizes the convergence versus the number of fitness evaluations with the different benchmarks including the genetic algorithm (GA) \cite{peng2021analysis}, and the DE \cite{huang2022placement}. 
The intial hyperparameters are  $\mathsf{F}_{\mathrm{init}} = 0.1$ and $\mathsf{CR}_{\mathrm{init}} = 0.9$. 
The maximum number of generations is $350$, which is equivalent to $35000$ fitness evaluations.}
\begin{figure*}[t]
    \centering
 \subfloat{\includegraphics[width=0.29\textwidth]{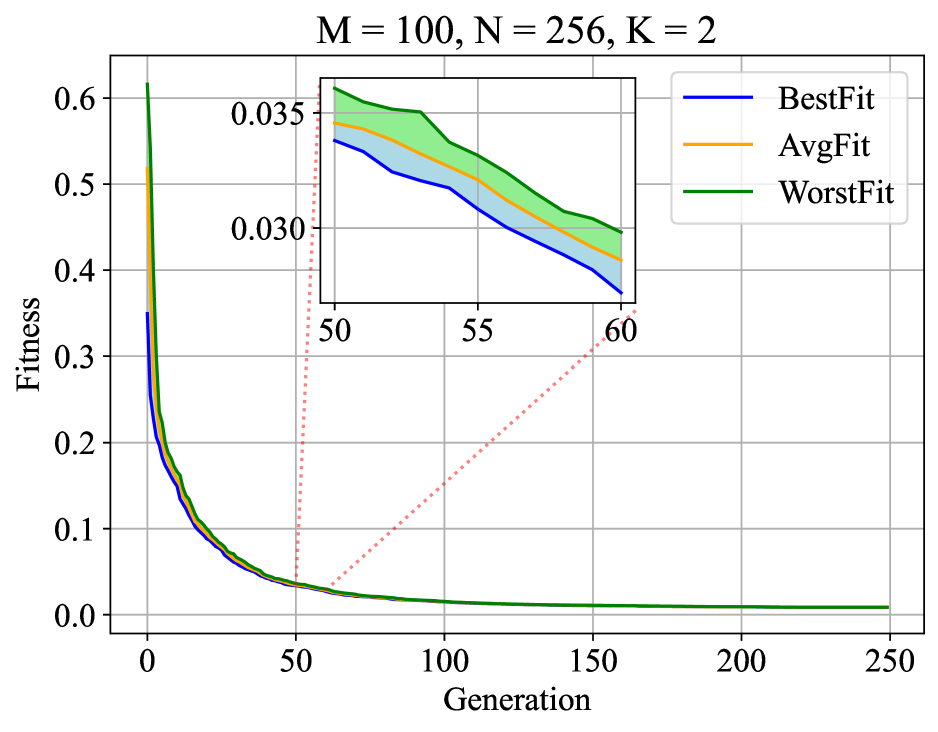}}
    \subfloat{\includegraphics[width=0.3\textwidth]{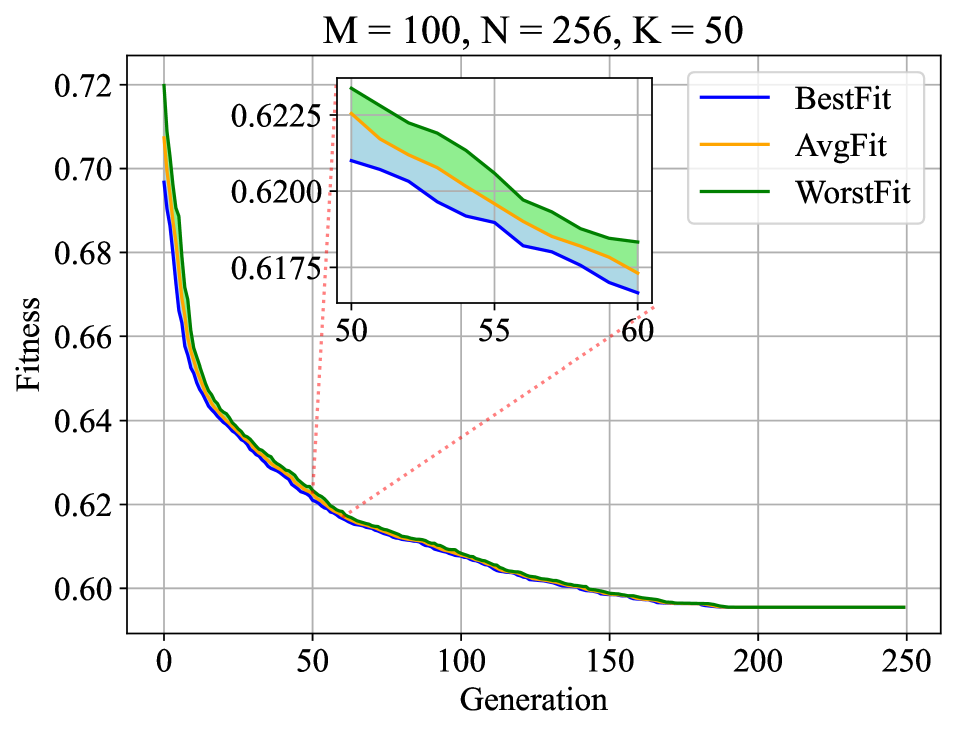}}
    \subfloat{\includegraphics[width=0.31\textwidth]{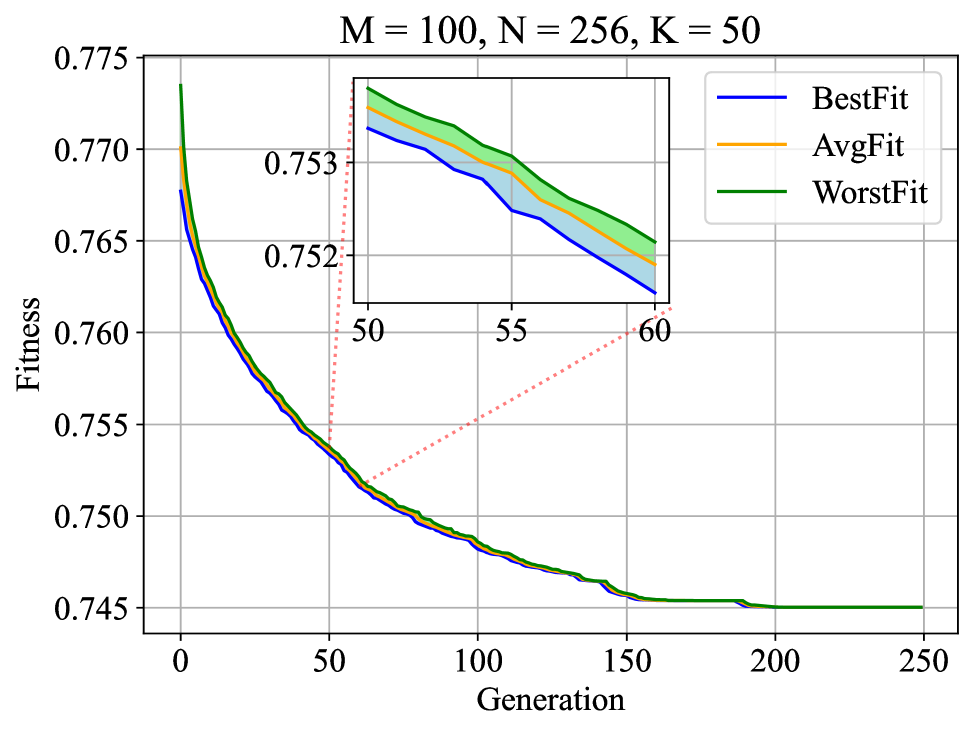}}
    \caption{Convergence of  Algorithm~\ref{alg:Improved_shade} as the behaviors of individuals (the best, worst, and average) with at $\rho/\sigma^2 = 5$ [dB].}
 \label{fig:gen_convergence_trend}
 \vspace{-0.5cm}
\end{figure*}
{The result shows that Algorithm~\ref{alg:Improved_shade} and the DE converge faster than the GA. In the first evaluations when the DE still performs well in the evolution of the solution, the local search step is not particularly effective. During these first evaluations, the improvement of Algorithm~\ref{alg:Improved_shade} over the DE is not significant. However, after about two thousand fitness evaluations, the individuals produced by the DE may be trapped in a local optimum, which is the so-called premature convergence. Consequently, the local search technique demonstrates significant effectiveness in helping the objective function value not converge prematurely. The GA does not get stuck in local minima early like the DE, the improvement of the best solution over the generations is slow. The numerical results point out a notable effectiveness of Algorithm~\ref{alg:Improved_shade} observed across all the considered network settings.} 

{Additionally, we examine the convergence of Algorithm~\ref{alg:Improved_shade} within the population and the stability level of the population, as illustrated in Fig~\ref{fig:gen_convergence_trend}. At each generation, we observe the fitness of the best individual, the worst individual, and all the individuals on average, which are denoted as `BestFit', `WorstFit', and  `AvgFit' respectively. The population is gradually stable over the generations for the three observed metrics. The difference in fitness among individuals decreasing along the generations indicates adaptability. It means that the less competent individuals are being gradually eliminated.}

{In Fig.~\ref{fig:optimization_result}, we compare the effectiveness of Algorithm~\ref{alg:Improved_shade} with the GA, DE, {gradient descent algorithm}, and the baseline that involves the random phase shift selection and the RZF technique. Algorithm~\ref{alg:Improved_shade} outperforms all the remaining benchmarks over the different network settings. Furthermore, the SER of each user is significantly improved as the signal-to-noise ratio increases. With a small number of users, Algorithm~\ref{alg:Improved_shade} shows substantially better performance than the others. For instance, with  $K=2$ users and $\rho/\sigma^2 = 0$~[dB], Algorithm~\ref{alg:Improved_shade} shows improvements of $39.4\%$, $59.4\%$, $63.4\%$, and $76.5\%$ compared to the GA, the DE, the gradient descent algorithm, and the random phase shift design together with the RZF technique, respectively.}

\begin{figure*}[t]
    \centering
    \subfloat{\includegraphics[width=0.3\textwidth]{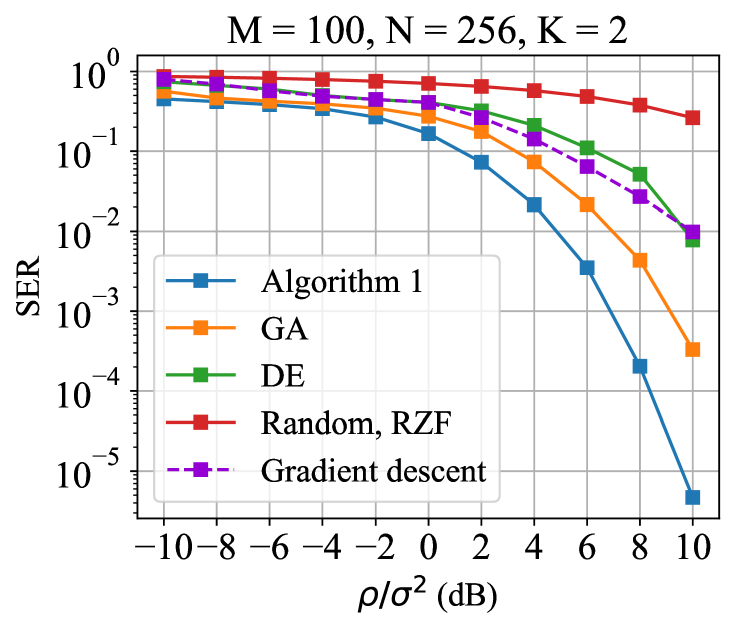}}
    \subfloat{\includegraphics[width=0.3\textwidth]{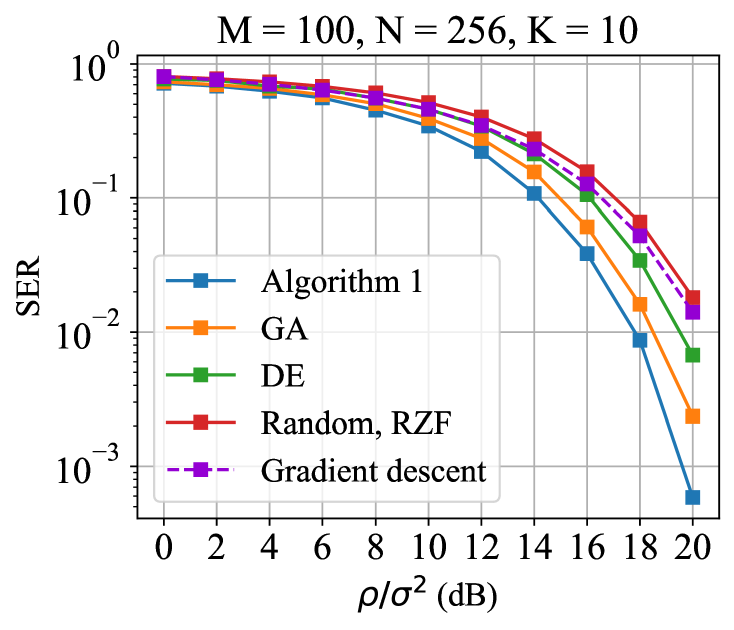}}
    \subfloat{\includegraphics[width=0.3\textwidth]{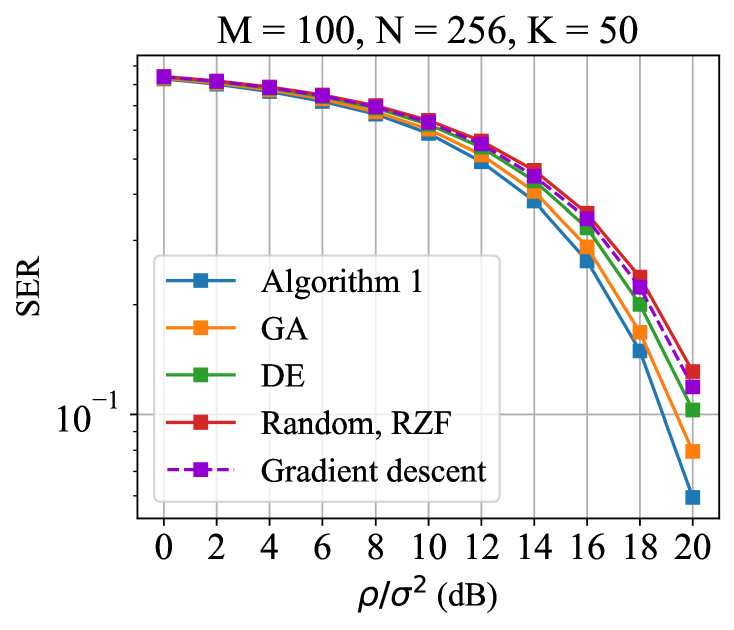}}
    \caption{{SER of different benchmarks comprising Algorithm~\ref{alg:Improved_shade}, the GA \cite{peng2021analysis}, the DE \cite{huang2022placement}, the gradient descent algorithm \cite{ye2020joint}, and the random phase shift together with the RZF technique.}}
    \label{fig:optimization_result}
    \vspace{-0.25cm}
\end{figure*}

\begin{figure}[http]
    \centering
    \includegraphics[width=0.85\linewidth]{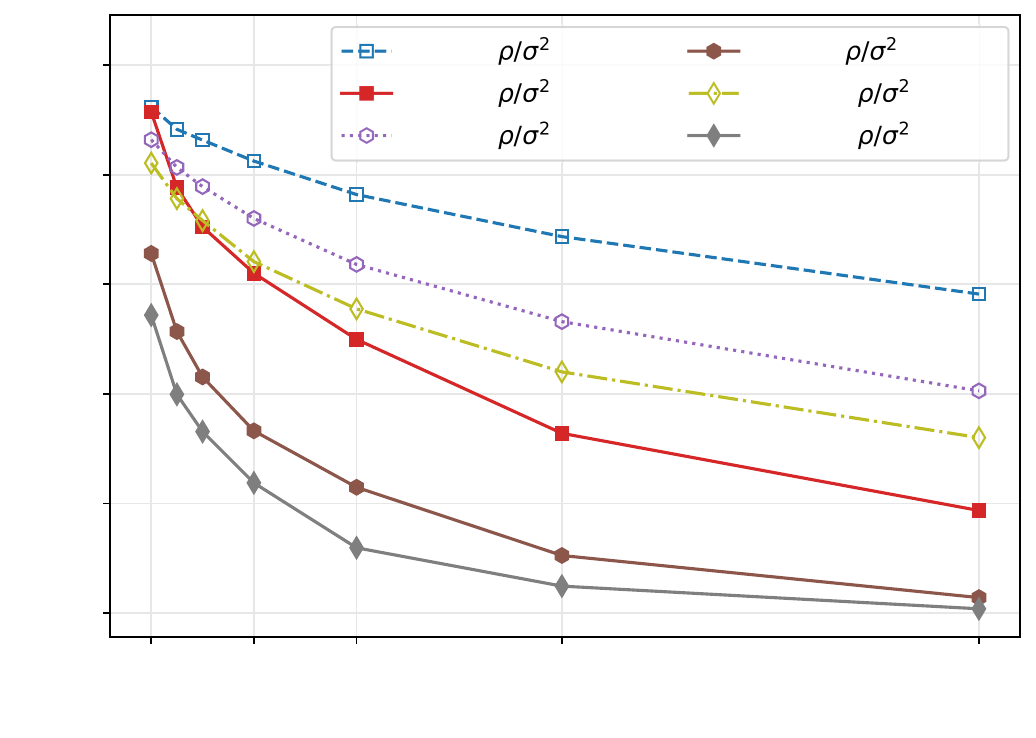}
    \caption{{SER versus the number of phase shift elements with the different number of  BS antennas serving $10$ users.}}
    \label{fig:impact_by_number_RIS_BS}
\end{figure}

{In Fig.~\ref{fig:impact_by_number_RIS_BS}, we illustrate the SER of the system serving $K = 10$ users under the different numbers of RIS reflecting elements and antennas at BS. The number of phase shift elements equal to zero corresponds to a system without the support of the RIS. Increasing the number of RIS elements with the fixed number of BS antennas decreases the SER of each user on average. With the same number of antennas and SINR, a system supported by $256$ phase shift elements performs up to $ 90\%$ better SER than without the presence of an RIS. It highlights the significant role of the RIS. After reaching a sufficiently large number of the phase shift elements, the improvement of the SER diminishes because it approaches the lower bound and the solution space extension becomes more challenging to optimize. Apart from this, lower SER can be achieved by equipping more antennas at the BS. We also investigate the system performance with $M \in \{10, 50, 100\}$ antennas, the same number of RIS elements, and  $\rho/\sigma^2 \in \{5, 15\}$~[dB]. Under similar conditions, the average SER per user significantly improves as the number of BS antennas increases.} 

\begin{figure}[t]
    \centering
\subfloat{\includegraphics[width=0.245\textwidth]{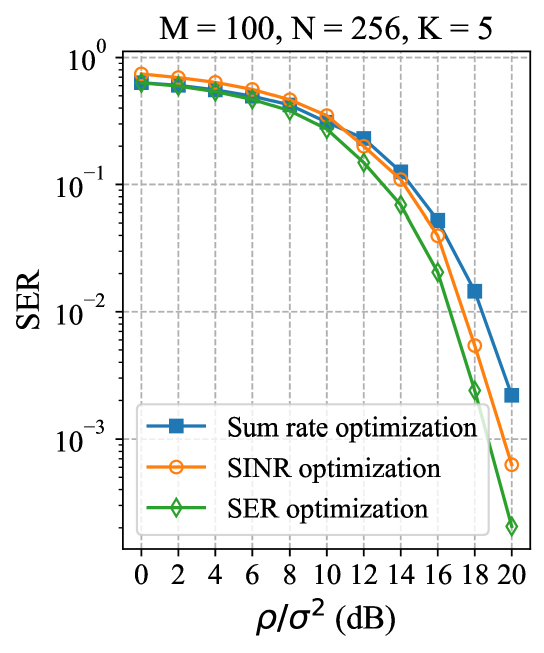}}
    \subfloat{\includegraphics[width=0.245\textwidth]{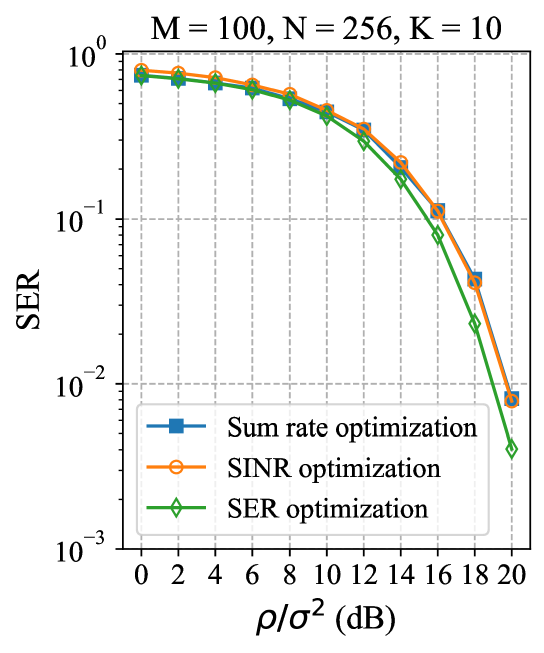}}
    \caption{SER of a RIS-assisted multiuser MIMO system with different objective functions consisting of the total SER,  {the minimum SINR}, and the sum-rate.}
    \label{fig:optimization_pro}
    \vspace{-0.35cm}
\end{figure}

{To leverage the benefits of the SER optimization in problem~\eqref{ProBlemv2}, we further investigate the two other objective functions comprising the sum rate, i.e., $\sum_{k=1}^K R_k$, and the minimum SINR among users, i.e., $\min\left\{ \mathrm{SINR}_k\right\}_{k=1}^K$, by modifying the proposed algorithm for the new fitness functions. Utilizing the SER, Fig.~\ref{fig:optimization_pro} illustrates the superior efficiency of the total SER obtained by solving problem~\eqref{ProBlemv2} over the remaining measurement metrics. Despite optimizing the sum rate and the minimum SINR, the SER can also improve for each user. However, there is still a gap between these benchmarks and the total SER minimization. Hence, we verify that problem~\eqref{ProBlemv2} could not be replaced by either the minimum SINR or the sum-rate maximization in finding the minimal SER.}

\begin{table*}[t]
\caption{SER of the RIS-assisted multiuser MIMO system based on the joint active and passive optimization and the passive beamforming design with the RZF.}
\label{tab:ser_compare_linear_precoding}
{ \centering \begin{tabular}{cc|ccccccccccc}
\hline
\multirow{2}{*}{\textbf{Scenario}} & \multirow{2}{*}{\textbf{Optimization}} & \multicolumn{11}{c}{$\pmb{\rho}/\pmb{\sigma}^2$}                                            \\ \cline{3-13} 
                          &                            & \multicolumn{1}{c|}{\textbf{0}} & \multicolumn{1}{c|}{\textbf{2}} & \multicolumn{1}{c|}{\textbf{4}} & \multicolumn{1}{c|}{\textbf{6}} & \multicolumn{1}{c|}{\textbf{8}} & \multicolumn{1}{c|}{\textbf{10}} & \multicolumn{1}{c|}{\textbf{12}} & \multicolumn{1}{c|}{\textbf{14}} & \multicolumn{1}{c|}{\textbf{16}} & \multicolumn{1}{c|}{\textbf{18}} & \textbf{20} \\ \hline
\multirow{2}{*}{2 users}  & passive/active             & \multicolumn{1}{c|}{0.6333}     & \multicolumn{1}{c|}{0.5916}     & \multicolumn{1}{c|}{0.542}      & \multicolumn{1}{c|}{0.4714}     & \multicolumn{1}{c|}{0.3837}     & \multicolumn{1}{c|}{0.2636}      & \multicolumn{1}{c|}{0.1551}      & \multicolumn{1}{c|}{0.0671}      & \multicolumn{1}{c|}{0.0205}      & \multicolumn{1}{c|}{0.0034}      & 0.0002      \\
                          & passive/linear precoding   & \multicolumn{1}{c|}{0.6626}     & \multicolumn{1}{c|}{0.6175}     & \multicolumn{1}{c|}{0.5637}     & \multicolumn{1}{c|}{0.4912}     & \multicolumn{1}{c|}{0.3981}     & \multicolumn{1}{c|}{0.2772}      & \multicolumn{1}{c|}{0.1639}      & \multicolumn{1}{c|}{0.071}       & \multicolumn{1}{c|}{0.0218}      & \multicolumn{1}{c|}{0.0036}      & 0.0002      \\ \hline
\multirow{2}{*}{10 users} & passive/active             & \multicolumn{1}{c|}{0.7425}     & \multicolumn{1}{c|}{0.7089}     & \multicolumn{1}{c|}{0.6626}     & \multicolumn{1}{c|}{0.6041}     & \multicolumn{1}{c|}{0.5204}     & \multicolumn{1}{c|}{0.4186}      & \multicolumn{1}{c|}{0.2967}      & \multicolumn{1}{c|}{0.178}       & \multicolumn{1}{c|}{0.0815}      & \multicolumn{1}{c|}{0.0242}      & 0.0038      \\  
                          & passive/linear precoding   & \multicolumn{1}{c|}{0.7648}     & \multicolumn{1}{c|}{0.7261}     & \multicolumn{1}{c|}{0.675}      & \multicolumn{1}{c|}{0.6136}     & \multicolumn{1}{c|}{0.529}      & \multicolumn{1}{c|}{0.4299}      & \multicolumn{1}{c|}{0.3043}      & \multicolumn{1}{c|}{0.182}       & \multicolumn{1}{c|}{0.0834}      & \multicolumn{1}{c|}{0.0247}      & 0.0038      \\ \hline
\multirow{2}{*}{50 users} & passive/active             & \multicolumn{1}{c|}{0.8354}     & \multicolumn{1}{c|}{0.8096}     & \multicolumn{1}{c|}{0.7773}     & \multicolumn{1}{c|}{0.7351}     & \multicolumn{1}{c|}{0.6809}     & \multicolumn{1}{c|}{0.6139}      & \multicolumn{1}{c|}{0.5277}      & \multicolumn{1}{c|}{0.4237}      & \multicolumn{1}{c|}{0.3026}      & \multicolumn{1}{c|}{0.1868}      & 0.0915      \\  
                          & passive/linear precoding   & \multicolumn{1}{c|}{0.8381}     & \multicolumn{1}{c|}{0.8129}     & \multicolumn{1}{c|}{0.7815}     & \multicolumn{1}{c|}{0.7407}     & \multicolumn{1}{c|}{0.688}      & \multicolumn{1}{c|}{0.6224}      & \multicolumn{1}{c|}{0.5369}      & \multicolumn{1}{c|}{0.4323}      & \multicolumn{1}{c|}{0.3094}      & \multicolumn{1}{c|}{0.1912}      & 0.0937      \\ \hline
\end{tabular}}
\vspace{-0.5cm}
\end{table*}
\begin{figure}[t]
    \centering   \includegraphics[width =0.6 \linewidth]{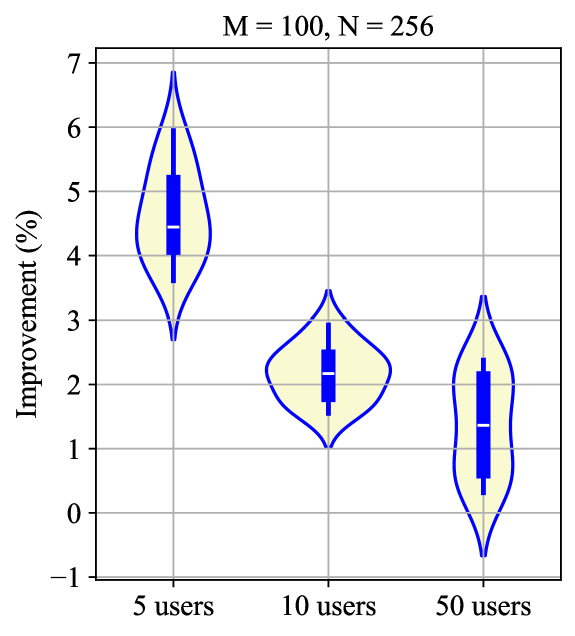}
    \caption{SER improvement based on joint active and passive optimization compared to passive beamforming design with the RZF precoding.}
    \label{fig:improvement_linear}
    \vspace{-0.5cm}
\end{figure}

We now compare the SER obtained by solving problems~\eqref{ProBlemv2} and \eqref{ProBlemv3} with $K = 5, 10, 50$ users. In particular, TABLE~\ref{tab:ser_compare_linear_precoding} presents the SER as a function of $\rho/\sigma^2$ by jointly optimizing the active and passive beamforming, along with passive beamforming optimization and the linear precoding, i.e., the RZF precoding. Meanwhile, Fig.~\ref{fig:improvement_linear} illustrates the enhancement of the average SER value by solving these corresponding optimization problems. The results obtained from solving problem~\eqref{ProBlemv2} are better than those of problem~\eqref{ProBlemv3} with the improvement from $1-6\%$, and decreasing as the number of users increases.

\begin{figure}[t]
    \centering   \includegraphics[width =0.8 \linewidth]{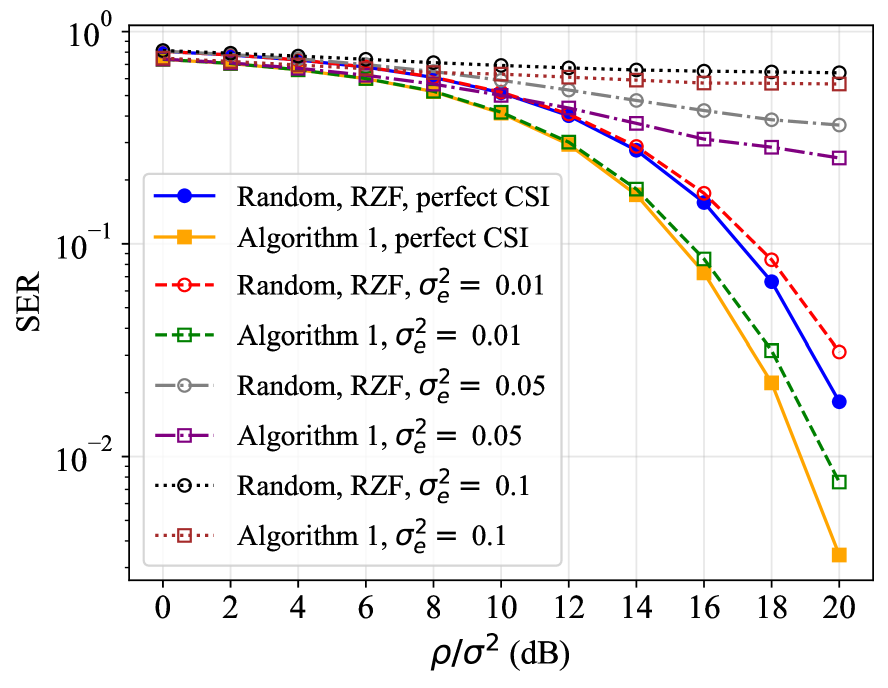}
    \caption{SER comparisons in presence of CSI estimation errors with $k = 10$ users.}
    \label{fig:CSI}
    \vspace{-0.5cm}
\end{figure}

{In reality, perfect instantaneous channels may be challenging due to the limited coherence intervals that result in estimation errors. Therefore, we evaluate the proposed algorithm in case of imperfect channel state information. The channel estimate of user~$k$ is given by $\tilde{\mathbf{z}}_k = \mathbf{z}_k + \mathbf{e}_k$, where $\mathbf{e}_k$ is the estimation error whose elements are distributed as $\mathcal{CN}(0, \sigma_e^2)$.
}
{Fig.~\ref{fig:CSI} demonstrates the SER per user with the different quality of channel estimation error, i.e., the variance $\sigma^2_e \in \{0.01, 0.05, 0.1\}$. As predicted, the SER increases with imperfect channel state information. However, there is still a performance gap between the proposed algorithm and the joint randomly phase-shift design and RZF in all the scenarios. With the small error variance $\sigma^2_e = 0.01$, the system shows its robustness. Nonetheless, greater differences occur as $\rho/\sigma^2$ increases from $16$~[dB] to $20$~[dB]. With large error variance $\sigma^2_e \in \{0.05, 0.1\}$, the degradation of the SER is more pronounced. The channel estimate differing significantly from the true channel demonstrates the superior performance of the proposed algorithm over the remaining benchmarks. Hence, it emphasizes the suitability of our proposed solution for practical applications.}


\begin{figure}
    \centering
    \includegraphics[width = 0.9 \linewidth]{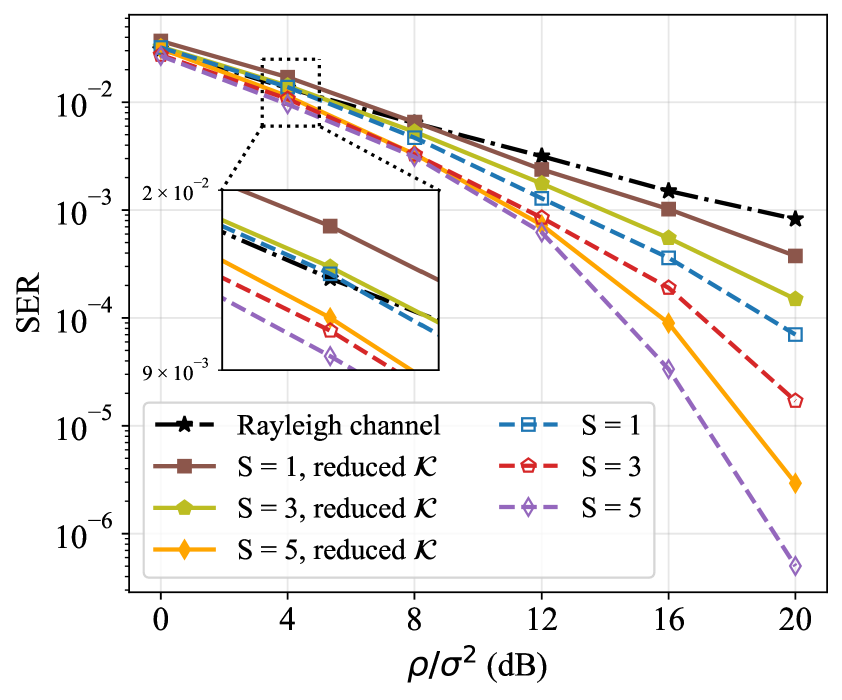}
    \caption{SER of the system with different numbers of specular components of the channel $S \in \{0, 1, 3, 5\}$, with and without a 9 dB reduction in Rician $\mathcal{K}$-factors.}
    \label{fig:impact_Rician}
\end{figure}

{We now analyze the influences of line of sight components (LoS) on the SER with presence of the RIS. The propagation channel  between the BS and user~$k$, the RIS and user~$k$, and the BS and the RIS, are respectively given  based on \cite{demir2022channel} as
\begin{align}
& \mathbf{u}_k = \sqrt{\frac{\mathcal{K}}{\mathcal{K} + 1}}\sum\nolimits_{s = 1}^{S_k^{\mathbf{u}}}e^{j\theta_{k, s}^{\mathbf{u}}} \overline{\mathbf{u}}_{k, s} + \sqrt{\frac{1}{\mathcal{K} + 1}}\tilde{\mathbf{u}}_k, \label{eq:uk} \\
& \mathbf{g}_k = \sqrt{\frac{\mathcal{K}}{\mathcal{K} + 1}}\sum\nolimits_{s = 1}^{S_k^{\mathbf{g}}}e^{j\theta_{k, s}^{\mathbf{g}}} \overline{\mathbf{g}}_{k, s} + \sqrt{\frac{1}{\mathcal{K} + 1}}\tilde{\mathbf{g}}_k, \\
& \mathbf{H} = \sqrt{\frac{\mathcal{K}}{\mathcal{K} + 1}}\sum\nolimits_{s = 1}^{S^{\mathbf{H}}}e^{j\theta^{\mathbf{H}}_{s}} \overline{\mathbf{H}}_s + \sqrt{\frac{1}{\mathcal{K} + 1}}\tilde{\mathbf{H}}, \label{eq:H}
\end{align}
where $S^{\mathbf{u}}_k$, $S^{\mathbf{g}}_k$, and $S^{\mathbf{H}}$ denote the number of specular components (the dominant paths) of these communication links. Note that the considered channel model admits Rayleigh fading with $S = S^{\mathbf{u}}_k = S^{\mathbf{g}}_k = S^{\mathbf{H}} = 0$. Let $\theta^{\mathbf{u}}_{k,s}$, $\theta^{\mathbf{g}}_{k,s}$ and $\theta^{\mathbf{H}}_s$ be the $s$-th specular components of the channels, which are independent uniformly distributed variables in $[0, 2\pi)$. If there exists one LOS path, then one of the specular components $e^{j\theta_{k, s}^{\mathbf{u}}} \overline{\mathbf{u}}_{k, s}$ ($e^{j\theta_{k, s}^{\mathbf{g}}} \overline{\mathbf{g}}_{k, s}$, or $e^{j\theta^{\mathbf{H}}_{s}} \overline{\mathbf{H}}_s$) corresponds to the path. The nonspecular parts $\tilde{\mathbf{u}}_k$, $\tilde{\mathbf{g}}_k$ and $\tilde{\mathbf{H}}$ represent the summation of other diffusely propagating multipath components. In \eqref{eq:uk}--\eqref{eq:H},  $\mathcal{K}$ is the Rician factor, with $\mathcal{K} = 13 - 0.03\Delta$ (dB), where $\Delta$ is the distance between the transmiter and the receiver. Fig.~\ref{fig:impact_Rician}  analyzes the SER per user with a different number of specular components, denoted as ${S}$, for all the channels. With ${S} = 1$, we assume the LoS path is the only specular component. Once $\mathcal{S} = 3, 5$, the original LoS channel gain is distributed randomly over the two non-LoS dominant components by keeping the power ratio of the LoS components at $0.5$.  The non-LoS specular components' angles of arrival and departure are randomly generated around the LOS angles with being less than $60^\circ$  and $15^\circ$ for the azimuth and elevation angle, respectively. The results compare the SER per user with either the Rician or Rayleigh channels. The SER reduces significantly with more specular components in the environment since the power is distributed among dominant paths. When there is only one LoS path, the SER slightly improves $41\%$ on average compared to the Rayleigh fading channels.  However, the improvement is more significant with $S \in \{ 3, 5\}$, i.e., $56\%$ and $61\%$, respectively. To observe the impact of the Rician factor, we consider a $9$ dB reduction that is introduced on the distance-dependent Rician $\mathcal{K}$-factors of the propagation channel between the RIS and each user, with $\mathcal{K} = 13 - 0.03\Delta - 9$ (dB). This scenario maintains the channel gain in the LOS case and redistributing it so that the specular components are less dominant compared to the nonspecular Rayleigh fading. Fig.~\ref{fig:impact_Rician} shows that the SER performance degrades when there is a reduction compared to the original Rician $\mathcal{K}$-factor.}

\vspace{-0.25cm}
\section{Conclusion}\label{Sec:Concl}
{This paper analyzed and optimized communication reliability in RIS-assisted MIMO systems {by considering the optimization problem of minimizing the average SER of users} through active and passive beamforming designs. By seamlessly combining active and passive beamforming strategies, the paper not only addressed the challenges posed by the propagation environments but also pioneered a method that intricately navigates the nuances of RIS integration. We derived the analytical SER for each user that can be applied for multi-user scenarios and under arbitrary fading channel models. Leveraging an improved DE algorithm, we further focused on minimizing the average downlink SER for modulated signals. The proposed algorithm emerged as a pivotal factor, showcasing its adaptability and efficiency in optimizing the critical metric in the assessment of communication reliability. The results obtained underscore the efficacy of the proposed evolutionary approach, shedding light on the potential it holds for RIS deployments in practice.}

\vspace{-0.25cm}
\appendix
\vspace{-0.25cm}
\subsection{Proof of Lemma~\ref{theorem:ser}}\label{appen:ser}
{The SER of user $k$ in \eqref{eq:SERk} represents the probability that the received signal $r_k$ in \eqref{eq:rk} does not fall within the Voronoi region associated with the transmitted signal $s_k$. According to the law of total probability, $\mathsf{SER}_k$ is reformulated as
\begin{equation}\label{eq:ser48}
\begin{split}
\mathsf{SER}_k = \frac{1}{m} \sum_{s_t \in \mathcal{M}} \mathsf{Pr}\left(r_k\notin\mathcal{V}(s_k) \big| s_k = s_t\right), 
\end{split}
\end{equation}
due to the equiprobable symbols. 
We stress that all the constellation points belonging to one subset are equiprobable. Plugging \eqref{eq:Zknuk} into \eqref{eq:rk}, the received signal $r_k$ at user~$k$ is redefined as
\begin{equation}\label{eq:rk_s_z_n}
r_k = s_k+ \zeta_k +\nu_k.
\end{equation}
By treating mutual interference as noise, which is shown in \eqref{eq:INkNOk}, we observe that $\zeta_k + \nu_k \sim \mathcal{CN}(0, \mathsf{IN}_k(\{ \mathbf{w}_k \}, \pmb{\Phi}) + \mathsf{NO}_k(\{ \mathbf{w}_k \}, \pmb{\Phi}))$. Consequently, both $\Re(\zeta_k + \nu_k)$ and $\Im(\zeta_k + \nu_k)$ are independent and identically distributed as
\begin{align}
& \Re(\zeta_k+\nu_k) \sim \mathcal{N} \left(0, 0.5(\mathsf{IN}_k(\{ \mathbf{w}_k \}, \pmb{\Phi}) + \mathsf{NO}_k(\{ \mathbf{w}_k \}, \pmb{\Phi})) \right), \label{eq:Re} \\
& \Im(\zeta_k+\nu_k) \sim \mathcal{N}  \left(0, 0.5(\mathsf{IN}_k(\{ \mathbf{w}_k \}, \pmb{\Phi}) + \mathsf{NO}_k(\{ \mathbf{w}_k \}, \pmb{\Phi}) \right). \label{eq:Im}
\end{align}
For ease of writing, we introduce a new variable $\mathsf{X}$ defined as
\begin{equation}\label{eq:defX}
\mathsf{X}=\operatorname{erfc}\left(\sqrt{3\mathsf{SINR}_k(\{ \mathbf{w}_k \}, \pmb{\Phi})/(m-1)}\right).
\end{equation}
Thanks to the Gaussian distributions, we obtain the  result
\begin{equation} \label{eq:prRe>=delta}
\begin{split}
    &\mathsf{Pr}(\Re(\zeta_k + \nu_k) \ge \delta)= \mathsf{Pr}(\Im(\zeta_k + \nu_k)\ge\delta)\\ &\stackrel{(a)}{=} \int_z^{+\infty} \frac{e^{-\frac{x^2}{\mathsf{IN}_k(\{ \mathbf{w}_k \}, \pmb{\Phi})+\mathsf{NO}_k(\{ \mathbf{w}_k \}, \pmb{\Phi})}}}{\sqrt{\pi(\mathsf{IN}_k(\{ \mathbf{w}_k \}, \pmb{\Phi})+\mathsf{NO}_k(\{ \mathbf{w}_k \}, \pmb{\Phi}))}}  \operatorname{d}x \\
    &\stackrel{(b)}{=} \dfrac{1}{2} \operatorname{erfc}\left(\dfrac{\delta\sqrt{2}}{\sqrt{\mathsf{IN}_k(\{ \mathbf{w}_k \}, \pmb{\Phi})+\mathsf{NO}_k(\{ \mathbf{w}_k \}, \pmb{\Phi})}}\right)
    \stackrel{(c)}{=} \dfrac{\mathsf{X}}{2},
\end{split}
\end{equation}
where $(a)$ is obtained by \eqref{eq:Re} and \eqref{eq:Im}; $(b)$ is obtained by exploiting the $\operatorname{erfc}(\cdot)$ function; and $(c)$ is obtained by using the SINR in \eqref{eq:sinr}. Furthermore, we can get another result
\begin{equation}
\begin{split}
    \mathsf{Pr}&(-\delta \le \Re(\zeta_k + \nu_k)<\delta)= \mathsf{Pr}(-\delta \le \Im(\zeta_k + \nu_k)<\delta)\\ 
    &=1- \mathsf{Pr}\big( -\delta>\Re(\zeta_k + \nu_k)\  \cup\ \delta \le \Re(\zeta_k + \nu_k)\big)\\
    &\stackrel{(a)}{=} 1- 2\mathsf{Pr}\big(\Re(\zeta_k + \nu_k) \ge\delta\big) = 1- \mathsf{X}.
\end{split}
\end{equation}
where $(a)$ is because of the identity $\mathsf{Pr}(X \cup Y) = \mathsf{Pr}(X) + \mathsf{Pr}(Y)$ for two mutually exclusive events $X, Y$ and the symmetric property of the Gaussian distribution.
We now consider the error probability for the three specifically modulated symbols $s_1, s_2,$ and $ s_3$. The error probability for the case $s_k = s_1$ is
\begin{equation}\label{eq:errorsk=s1}
\begin{split}
    \mathsf{Pr}&\big(r_k\notin \mathcal{V}(s_k) \big| s_k = s_1\big) = 1 - \mathsf{Pr}\big(r_k \in \mathcal{V}(s_1)\big| s_k = s_1\big)\\
        &\stackrel{(a)}{=} 1 - \mathsf{Pr}\Big(\Re(r_k) \ge \delta(\sqrt{m}-2) \\
        & \qquad \&\ \Im(r_k) \ge \delta(\sqrt{m}-2)\ \Big|\ s_k = s_1 \Big)\\
        &\stackrel{(b)}{=} 1 - \mathsf{Pr}\left(\Re(r_k) \ge \delta(\sqrt{m} - 2)\ \Big|\ s_k = s_1 \right)  \\ 
        & \qquad \quad \times \mathsf{Pr}\left(\Im(r_k) \ge \delta(\sqrt{m}-2)\ \Big|\ s_k = s_1 \right) \\
        &\stackrel{(c)}{=} 1 - \mathsf{Pr}\left( \delta(\sqrt{m} - 1) + \Re(\zeta_k+\nu_k) \ge \delta(\sqrt{m} - 2)\right)  \\
        & \qquad  \times \mathsf{Pr}\left(\delta(\sqrt{m} - 1) + \Im(\zeta_k+\nu_k) \ge \delta(\sqrt{m}-2)\right)\\
        &\ =\ 1 - \mathsf{Pr}\left(\Re(\zeta_k+\nu_k) \ge -\delta\right) \mathsf{Pr}\left(\Im(\zeta_k+\nu_k) \ge -\delta\right)\\
        &= 1 - \left( 1 - \mathsf{Pr}(\Re(\zeta_k+\nu_k) \ge \delta)\right) \left( 1 - \mathsf{Pr}(\Im(\zeta_k+\nu_k) \ge \delta)\right)\\
        &\stackrel{(d)}{=} 1 - \left( 1 - \frac{1}{2}\mathsf{X} \right)^2 = \mathsf{X} - \frac{1}{4}\mathsf{X}^2,
\end{split}\end{equation}
where $(a)$ is obtained from \eqref{eq:voronois_1}; $(b)$ is obtained by the identity $\mathsf{Pr}(X\& Y) = \mathsf{Pr}(X)\mathsf{Pr}(Y)$ for the two independent variables $X, Y$; $(c)$ is obtained by using \eqref{eq:rk_s_z_n} with $s_k= s_1$; and $(d)$ is obtained by using \eqref{eq:prRe>=delta}. We compute the error probability of $s_k= s_2$ as 
\begin{equation}\label{eq:errorsk=s2}
\begin{split}
\mathsf{Pr}\big(r_k&\notin \mathcal{V}(s_k) \big| s_k = s_2\big)= 1 - \mathsf{Pr}\big(r_k \in \mathcal{V}(s_2)\big| s_k = s_2\big)\\
    &= 1 - \left(1 - \frac{1}{2}\mathsf{X}\right)\left(1 -\mathsf{X}\right) = \frac{3}{2}\mathsf{X} - \frac{1}{2} \mathsf{X}^2.
\end{split}
\end{equation}
We compute  the error probability for the case $s_k= s_3$ as 
\begin{equation}\label{eq:errorsk=s3}
\begin{split}
    \mathsf{Pr}\big(r_k&\notin \mathcal{V}(s_k) \big| s_k = s_3\big) = 1 - \mathsf{Pr}\big(r_k \in \mathcal{V}(s_3)\big| s_k = s_3\big)\\
    &= 1 - \left(1 -  \mathsf{X}\right)\left(1 -\mathsf{X}\right) = 2\mathsf{X} -  \mathsf{X}^2.
\end{split}\end{equation}
Using the obtained results in \eqref{eq:errorsk=s1}--\eqref{eq:errorsk=s3}, we now can derive the analytical SER of user $k$ as
\begin{equation} \label{eq:SERkappendix}
\begin{split}
   & \mathsf{SER}_k  \stackrel{(a)}{=} \frac{1}{m}\sum_{s_t\in \mathcal{M}} \mathsf{Pr}\left(r_k\notin\mathcal{V}(s_k) \big| s_k = s_t\right) \stackrel{(b)}{=} \\ 
    &  \frac{1}{m} \Big(4\mathsf{X}-\mathsf{X}^2 +6\left(\sqrt{m}-2\right)\mathsf{X}-2\left(\sqrt{m}-2\right) \mathsf{X}^2 + 2\left(\sqrt{m}-2\right)^2  \\
    &  \times \mathsf{X}-\left(\sqrt{m}-2\right)^2\mathsf{X}^2\Big) =2\left(1-\frac{1}{\sqrt{m}}\right)\mathsf{X}-\left(1-\frac{1}{\sqrt{m}}\right)^2\mathsf{X}^2,
\end{split}
\end{equation}
where $(a)$ is obtained by using \eqref{eq:ser48} and $(b)$ is the substitution from the results in  \eqref{eq:errorsk=s1}--\eqref{eq:errorsk=s3}. Finally, by exploiting \eqref{eq:defX} into \eqref{eq:SERkappendix}, we obtain the SER in the theorem.}
\vspace{-0.25cm}
\subsection{Proof of Theorem~\ref{theorem:local_convergence}}\label{appen:local_convergence}
{We define $\mathcal{P}^{(G)} = \left\{\mathbf{x}^{(1G)}, \ldots, \mathbf{x}^{(iG)},\ldots, \mathbf{x}^{(IG)}\right\}$ the population generated by Algorithm~\ref{alg:Improved_shade} at the $G$-th generation. After that, the convergence in the probability of the proposed algorithm is
\begin{equation}
    \lim_{G \rightarrow \infty}\mathsf{Pr}\left(\mathcal{P}^{(G)} \cap \mathcal{S}^{\ast}_{\varepsilon} \neq \varnothing \right) = 1,
\end{equation}
where $\mathcal{S}^{\ast}_{\varepsilon}$ is given in \eqref{eq:Sep}. For individual $\mathbf{x}^{(iG)}$, Algorithm \ref{alg:Improved_shade} generates $\pmb{\omega}^{(iG)}$ by utilizing mutation and crossover operators \eqref{eq:mutation} and \eqref{eq:crossover} with the density function expressed as
\begin{equation}
    \Omega\left(\pmb{\omega}^{(iG)}\right) = \begin{cases}
        1, & \mbox{if } \pmb{\omega}^{(iG)} \in \mathcal{S},\\
        0, & \mbox{otherwise}.
    \end{cases}
\end{equation}
Therefore, the probability of $\pmb{\omega}^{(iG)}$ belonging to $\mathcal{S}_{\varepsilon}^{*}$ is 
\begin{equation} \label{eq:PrOmega}
    \mathsf{Pr}\left(\pmb{\omega}^{(iG)} \in \mathcal{S}^{*}_{\varepsilon}\right) = \int_{\mathcal{S}_{\varepsilon}^{*}} \Omega\left(\pmb{\omega}^{(iG)}\right)\mathrm{d}x = \mu_1\left(\mathcal{S}_{\varepsilon}^{\ast}\right),
\end{equation}
where $\mathrm{d}x$ represents the differential elements comprising the phase shifts and beamforming coefficients associated with $\mathcal{S}_{\varepsilon}^{\ast}$. If  $\mathcal{B}^{(G)}$ is the set of individuals after performing the selection on the population $\mathcal{P}^{(G)}$, the probability, which $\mathcal{B}^{(G)}$ does not contain any individuals belonging to the set of the optimal solutions, is
\begin{equation}  \label{eq:PrB}
\begin{split}
    \mathsf{Pr}\left(\mathcal{B}^{(G)} \cap \mathcal{S}^{\ast}_{\varepsilon} = \varnothing\right) & \leq \left(1 - \mathsf{Pr}\left(\pmb{\omega}^{(iG)} \in \mathcal{S}^{\ast}_{\varepsilon}\right)P_{\mathrm{ep}}\right)^{I} \\
    & \stackrel{(a)}{=} \left(1 - \mu_1 (\mathcal{S}^{\ast}) P_{\mathrm{ep}}\right)^{I},
\end{split}
\end{equation}
where $(a)$ is obtained by the measure $\mu_1(\mathcal{S}_\varepsilon^\ast)$ in \eqref{eq:PrOmega}. In \eqref{eq:PrB}, we see that the diversity of the population will gradually enhance as $P_{\mathrm{ep}}$ increases.}

{After the selection,  we denote $\mathcal{Q}^{(G)}$ the set of individuals obtained by performing the local search on $\mathcal{B}^{(G)}$. We recall that the execution of local search is performed in \eqref{eq:neighbor_solution} with $\tilde{\pmb{\omega}}^{(iG)} \in \mathcal{Q}^{(G)}$, and $\pmb{\omega}^{(iG)} \in \mathcal{B}^{(G)}$. The probability density function of $\pmb{\xi}^{(iG)}$ is
\begin{equation} \label{eq:Omega}
   \begin{split}
           \Omega\left(\pmb{\xi}^{(iG)}\right) &= \left(\frac{1}{\tilde{\sigma}\sqrt{2\pi}}\right)^{N + 2MK} \prod_{n = 1}^{N+2MK}e^{\frac{-\left(\xi_{n}^{(iG)}\right)^2}{2\tilde{\sigma}^2}} \\
           & =\left(\frac{1}{\tilde{\sigma}\sqrt{2\pi}}\right)^{N + 2MK}e^{\sum_{n = 1}^{N + 2MK}\frac{-\left(\xi_{n}^{(iG)}\right)^2}{2\tilde{\sigma}^2}},
   \end{split} 
\end{equation}
which is a multivariate function of $N + 2MK$ Gaussian random variables. In this case, the probability of $\tilde{\pmb{\omega}}^{(iG)}$ belonging to $\mathcal{S}^{\ast}_{\varepsilon}$ is derived as follows
\begin{equation}    
\mathsf{Pr}\left(\tilde{\pmb{\omega}}^{(iG)} \in \mathcal{S}^{\ast}_{\varepsilon}\right) = \int_{\tilde{\mathcal{S}}^{\ast}_{\varepsilon}}\Omega(\pmb{\xi}^{(iG)})\text{d}\xi = \mu_2 (\tilde{\mathcal{S}}^{\ast}),
    \label{eq:probability_belong}
\end{equation}
where $\tilde{\mathcal{S}}^{\ast}_{\varepsilon}  = \{\hat{\mathbf{x}} - \pmb{\omega}^{(iG)} | \, \hat{\mathbf{x}} \in \mathcal{S}^{\ast}_{\varepsilon}\}$ and $\mathrm{d}\xi$ represents the differential elements consisting of the phase shifts and beamforming coefficients associated with the measure $\xi$ of $\tilde{\mathcal{S}}^{*}_{\varepsilon}$. Because of each element of individual $\pmb{\omega}^{(iG)}$ in the range of $[-1, 1]$, one observes that each element $   \tilde{\pmb{\omega}}^{(iG)} \in  \tilde{\mathcal{S}}^{\ast}_{\varepsilon}$ satisfies 
$[\tilde{\pmb{\omega}}^{(iG)}]_n \in [-2,2], n = \{1, \ldots, N + 2MK\}$. Consequently, the probability density function $ \Omega\left(\pmb{\xi}^{(iG)}\right)$ is lower bounded by
\begin{equation} \label{eq:Omegav1}
    \Omega\left(\pmb{\xi}^{(iG)}\right) \geq \left(\frac{1}{\tilde{\sigma}\sqrt{2\pi}}\right)^{N + 2 MK}e^{\frac{-2(N+2MK)}{\tilde{\sigma}^2}}, \forall \pmb{\xi}^{(iG)} \in \tilde{\mathcal{S}}^{\ast}_{\varepsilon}.
\end{equation}
Exploiting \eqref{eq:Omega} and \eqref{eq:Omegav1} into  \eqref{eq:probability_belong}, we attain the improvement in the probability of the local search is
\begin{equation} \label{eq:PrGl}
    \begin{split}
    &\mathsf{Pr}\left(\tilde{\pmb{\omega}}^{(iG)} \in \mathcal{S}_{\varepsilon}^{\ast} | \mathcal{B}^{(G)} \cap \mathcal{S}_{\varepsilon}^{\ast} = \varnothing\right)\\ 
    & \geq \mu_2\left(\tilde{\mathcal{S}}^{\ast}_{\varepsilon}\right)\left(\frac{1}{\tilde{\sigma}\sqrt{2\pi}}\right)^{N+2MK}e^{\frac{-2(N+2MK)}{\tilde{\sigma}^2}}, \forall \pmb{\xi}^{(iG)} \in \tilde{\mathcal{S}}^{\ast}_{\varepsilon}.
    \end{split}
\end{equation}
Due to the fact that $\mu_2(\mathcal{S}^{\ast}_{\varepsilon}) = \mu_2 (\tilde{\mathcal{S}}^{\ast}_{\varepsilon})$,  \eqref{eq:PrGl} is reformulated as 
\begin{equation}
    \begin{split}
    &\mathsf{Pr}\left(\tilde{\pmb{\omega}}^{(iG)} \in \mathcal{S}_{\varepsilon}^{\ast} | \mathcal{B}^{(G)} \cap \mathcal{S}_{\varepsilon}^{*} = \varnothing\right) \\
    & \geq \mu_2\left(\mathcal{S}^{\ast}_{\varepsilon}\right)\left(\frac{1}{\tilde{\sigma}\sqrt{2\pi}}\right)^{N + 2MK}e^{\frac{-2(N + 2MK)}{\tilde{\sigma}^2}}, 
    \end{split}
\end{equation}
which holds for $\forall \pmb{\xi}^{(iG)} \in \tilde{\mathcal{S}}^{\ast}_{\varepsilon}$. Conditioned on the failure of the canonical DE, the probability that the local search also failed to find the optimal solution in Algorithm~\ref{alg:Improved_shade} is
\begin{multline} \label{eq:PrLocal1}
\mathsf{Pr}\left(\mathcal{Q}^{(G)} \cap \mathcal{S}_{\varepsilon}^{\ast} = \varnothing | \mathcal{B}^{(G)} \cap \mathcal{S}_{\varepsilon}^{\ast} = \varnothing\right) \le \\ \left(1 - \mu_2\left(\mathcal{S}_{\varepsilon}^{*}\right)\left(\frac{1}{\tilde{\sigma}\sqrt{2\pi}}\right)^{N+2MK}e^{\frac{-2(N+2MK)}{\tilde{\sigma}^2}}\right)^{\tilde{\lambda}}.
\end{multline}
Therefore, the probability that the population contains the optimal solution at the $G$-th generation is computed as
\begin{equation} \label{eq:Probv1}
\begin{split}
         &\mathsf{Pr}\left(\mathcal{Q}^{(G)} \cap \mathcal{S}_{\varepsilon}^{\ast} \neq \varnothing \right) = 1 - \mathsf{Pr}\left(\mathcal{Q}^{(G)} \cap \mathcal{S}_{\varepsilon}^{\ast} = \varnothing \right) \stackrel{(a)}{=} 1  -  \\
         &  \mathsf{Pr}\left(\mathcal{Q}^{(G)} \cap \mathcal{S}_{\varepsilon}^{\ast} = \varnothing | \mathcal{B}^{(G)} \cap \mathcal{S}_{\varepsilon}^{\ast} = \varnothing\right) \mathsf{Pr}\left(\mathcal{B}^{(G)} \cap \mathcal{S}^{\ast}_{\varepsilon} = \varnothing\right),
\end{split}
\end{equation}
where $(a)$ is obtained by the Bayes' theorem. By exploiting \eqref{eq:PrB} and \eqref{eq:PrLocal1} into \eqref{eq:Probv1}, we obtain the result as in the theorem.}
\subsection{Proof of Corollary~\ref{corollary:converge}} \label{Appdixcorollary:converge}
{By virtue of Theorem~\ref{theorem:local_convergence}, one can get the probability that all the individuals after performing the reproduction and local search on the $G$-th population do not belong to $\mathcal{S}^\ast_{\varepsilon}$ as
    \begin{equation} \label{eq:Prv1}
        \mathsf{Pr}\left(\mathcal{Q}^{(G)} \cap \mathcal{S}^*_{\varepsilon} = \varnothing\right) \le 1 - \psi^{(G)},
    \end{equation}
    where $\psi^{(G)}$ is defined 
    for the $G$-th population as
    \begin{multline}
        \psi^{(G)} = 1- \left(1 - \mu_1\left(\mathcal{S}_{\varepsilon}^{\ast}\right)P_{\mathrm{ep}} \right )^{I} \times \\ \left(1 - \mu_2\left(\mathcal{S}_{\varepsilon}^{\ast}\right)\left(\frac{1}{\tilde{\sigma}\sqrt{2\pi}}\right)^{N+2MK}e^{\frac{-2(N + 2MK)}{\tilde{\sigma}^2}}\right)^{{\lambda}_G},
    \end{multline}
   and $\lambda_G$ is the number of individuals in the $G$-th generation. We indeed observe that $\sum_{G=1}^{\infty} \psi^{(G)}$ diverges, i.e., $\sum_{G=1}^{\infty} \psi^{(G)} \rightarrow \infty$ as $G \rightarrow \infty$. Since the selection of Algorithm~\ref{alg:Improved_shade} is the elitist mechanism, the best individual will remain in the population through generations. Hence, the probability that $\mathcal{P}^{(G)}$ does not contain the $\varepsilon$-optimal solution to problem~\eqref{ProBlemv2} is 
    \begin{equation}\label{eq:evaluate_psiG}
     \begin{split}
        \mathsf{Pr}\left(\mathcal{P}^{(G)}\cap \mathcal{S}^*_{\varepsilon} = \varnothing \right) &= \prod_{g=1}^{G-1} \mathsf{Pr}\left(\mathcal{Q}^{(g)}\cap \mathcal{S}^*_{\varepsilon} = \varnothing \right) \\
        &\le \prod_{g=1}^{G-1} \left(1- \psi^{(g)}\right), 
    \end{split}
    \end{equation}
where the last inequality in \eqref{eq:evaluate_psiG} is obtained based on \eqref{eq:Prv1}.  We further observe the following property
\begin{equation} \label{eq:limprodpsiG}
        \begin{split}
            & 0<\prod_{g=1}^{G-1} \left(1- \psi^{(g)}\right) = e^{\log \left( \prod_{g=1}^{G-1} \left(1- \psi^{(g)} \right) \right)} = e^{\sum_{g=1}^{G-1} \log \left(1- \psi^{(g)}\right)}\\ &  \stackrel{(a)}{\le} e^ {-\sum_{g=1}^{G-1} \psi^{(g)}} \rightarrow e^{-\infty}=0 \text{ as } G\rightarrow \infty,
        \end{split}
    \end{equation} where $(a)$ is obtained by the identity $\log (1-x) +x \le 0 \ \forall x \in [0,1)$. Consequently, we   obtain the following result
    \begin{equation}
        \begin{split}
            \lim_{G\rightarrow \infty}&\ \mathsf{Pr} \left(\mathcal{P}^{(G)} \cap \mathcal{S}^*_{\varepsilon} \neq \varnothing \right)= 1- \lim_{G\rightarrow \infty}\ \mathsf{Pr}\left(\mathcal{P}^{(G)} \cap \mathcal{S}^*_{\varepsilon} = \varnothing \right)\\ & \stackrel{(a)}{\geq} 1- \lim_{G\rightarrow \infty} \prod_{g=1}^{G-1} \left(1- \psi^{(g)}\right) \stackrel{(b)}{=} 
            1,
        \end{split}
    \end{equation}
    where $(a)$ is based on \eqref{eq:evaluate_psiG} and $(b)$ relies on \eqref{eq:limprodpsiG} together with the sandwich theorem. We therefore complete the proof.}

\bibliographystyle{IEEEtran}
\bibliography{references}

\end{document}